\documentclass[english,final, twocolumn]{IEEEtran}
\usepackage[T1]{fontenc}
\usepackage[latin9]{inputenc}
\usepackage{mathtools}
\usepackage{bm}
\usepackage{amsmath}
\usepackage{amsthm}
\usepackage{amssymb}
\usepackage{graphicx}

\makeatletter

\theoremstyle{plain}
\newtheorem{prop}{\protect\propositionname}
\newtheorem*{remark}{\protect\remarkname}
\usepackage{psfrag}
\usepackage{cite}
\usepackage{subfigure}\usepackage{algorithm}
\usepackage{algpseudocode}
\usepackage{bm}          

\author{
	Hao~Sun, Shicong~Liu,~\IEEEmembership{Graduate Student Member,~IEEE},~Xianghao~Yu,~\IEEEmembership{Senior Member,~IEEE},~Ying~Sun,~\IEEEmembership{Member,~IEEE},~and Liu Cao,~\IEEEmembership{Member,~IEEE}
	\thanks{Hao Sun, Shicong Liu, and Xianghao Yu are with the Department of Electrical Engineering, City University of Hong Kong, Hong Kong (e-mail:
		hao.sun@cityu.edu.hk; sc.liu@my.cityu.edu.hk; alex.yu@cityu.edu.hk).}\thanks{Ying Sun is with the School of Electrical Engineering and
		Computer Science, Pennsylvania State University, State College, PA 16801
		USA (e-mail: ybs5190@psu.edu).}\thanks{Liu Cao is with the
		City University of Hong Kong (Dongguan), Dongguan 523808, China,
		and also with the City University of Hong Kong, Hong Kong (e-mail:
		liu.cao@cityu-dg.edu.cn).}
	\thanks{A preliminary version of this work appeared in Proc. IEEE ICASSP 2026~\cite{SunLiuYuXSun:C26}. This journal version substantially extends the conference version by developing a time-truncated PnP-flow reconstruction algorithm with an ODE-method interpretation, refining the active trajectory-planning formulation, and providing significantly expanded experimental validation.}}

\usepackage[acronym]{glossaries}
\newcommand{\newac}{\newacronym}
\newcommand{\ac}{\gls}
\newcommand{\Ac}{\Gls}

\makeglossaries

\newac{speb}{SPEB}{square position error bound}
\newac[plural=EFIMs,firstplural=Fisher information matrices (EFIMs)]{efim}{EFIM}{Fisher information matrix}
\newac{ne}{NE}{Nash equilibrium}
\newac{mse}{MSE}{mean squared error}
\newac{toa}{TOA}{time of arrival}
\newac{tdoa}{TDOA}{time difference of arrival}
\newac{snr}{SNR}{signal-to-noise ratio}
\newac{lan}{LAN}{local area network}
\newac{psd}{PSD}{positive semidefinite}
\newac{pd}{PD}{positive definite}
\newac{wrt}{w.r.t.}{with respect to}
\newac{lhs}{L.H.S.}{left hand side}
\newac{wp1}{w.p.1}{with probability 1}
\newac{kkt}{KKT}{Karush-Kuhn-Tucker}
\newac{wlog}{w.l.o.g.}{without loss of generality}
\newac{mle}{MLE}{maximum likelihood estimation}
\newac{rssi}{RSSI}{received signal strength indicator}
\newac{mimo}{MIMO}{multiple-input multiple-output}
\newac{csi}{CSI}{channel state information}
\newac{fdd}{FDD}{frequency division duplexing}
\newac{ms}{MS}{mobile station}
\newac{bs}{BS}{base station}
\newac{bss}{BSs}{base stations}
\newac{d2d}{D2D}{device-to-device}
\newac{slnr}{SLNR}{signal-to-interference-leakage-and-noise-ratio}
\newac{ula}{ULA}{uniform linear antenna array}
\newac{pas}{PAS}{power angular spectrum}
\newac{mmse}{MMSE}{minimum mean square error}
\newac{zf}{ZF}{zero-forcing}
\newac{rzf}{RZF}{regularized zero-forcing}
\newac{as}{AS}{angular spread}
\newac{ue}{UE}{user equipment}
\newac{ues}{UEs}{user equipments}
\newac{iid}{i.i.d.}{independent and identically distributed} 
\newac{sinr}{SINR}{signal-to-interference-and-noise ratio}
\newac{tdd}{TDD}{time-division duplex}
\newac{rvq}{RVQ}{random vector quantization}
\newac{rhs}{R.H.S.}{right hand side}
\newac{mrc}{MRC}{maximum ratio combining}
\newac{cdf}{CDF}{cumulative distribution function}
\newac{a.s.}{a.s.}{almost surely}
\newac{los}{LOS}{line-of-sight}
\newac{jsdm}{JSDM}{joint spatial division and multiplexing}
\newac{map}{MAP}{maximum a posteriori}
\newac{klt}{KLT}{Karhunen-Lo\`eve Transform}
\newac{lbe}{LBE}{link bargaining equilibrium}
\newac{se}{SE}{Stackelberg equilibrium}
\newac{uav}{UAV}{unmanned aerial vehicle}
\newac{uavs}{UAVs}{unmanned aerial vehicles}
\newac{nlos}{NLOS}{non-line-of-sight}
\newac{pdf}{PDF}{probability density function}
\newac{em}{EM}{expectation-maximization}
\newac{knn}{KNN}{$k$-nearest neighbors}
\newac{svd}{SVD}{singular value decomposition}
\newac{nmf}{NMF}{non-negative matrix factorization}
\newac{umf}{UMF}{unimodality-constrained matrix factorization}
\newac{rmse}{RMSE}{rooted mean squared error}
\newac{olos}{OLOS}{obstructed line-of-sight}
\newac{mmw}{mmW}{millimeter wave}
\newac{ber}{BER}{bit error rate}
\newac{rss}{RSS}{received signal strength}
\newac{lp}{LP}{linear program}
\newac{ufw}{U-FW}{unimodal Frank-Wolfe}
\newac{utf}{UTF}{unimodality-constrained tensor factorization}
\newac{fw}{FW}{Frank-Wolfe}
\newac{iot}{IoT}{Internet-of-Things}
\newac{mae}{MAE}{mean absolute error}
\newac{crb}{CRB}{Cram\'er-Rao bound}
\newac{aoa}{AoA}{angle of arrival}
\newac{wcl}{WCL}{weighted centroid localization}
\newac{slf}{SLF}{spatial loss function}
\newac{btd}{BTD}{block-term tensor decomposition}
\newac{tps}{TPS}{thin plate spline}
\newac{nmse}{NMSE}{normalized mean squared error}
\newac{svt}{SVT}{singular value thresholding}
\newac{vae}{VAE}{variational autoencoder}
\newac{gan}{GAN}{generative adversarial network}
\newac{aod}{AOD}{angle of departure}
\newac{rbf}{RBF}{radial basis function}
\newac{loocv}{LOOCV}{leave one out cross validation}
\newac{lpr}{LPR}{local polynomial regression}
\newac{xl-mimo}{XL-MIMO}{extremely large multiple-input multiple-output}
\newac{6g}{6G}{sixth-generation}
\newac{cpd}{CPD}{canonical polyadic decomposition}
\newac{nnm}{NNM}{nuclear norm minimization}
\newac{mimo-ofdm}{MIMO-OFDM}{multiple-input multiple-output orthogonal frequency division multiplexing}
\newac{mimo-cdma}{MIMO-CDMA}{multiple-input multiple-output code division multiple access}
\newac{pad}{PAD}{power-angular-delay}
\newac{hmm}{HMM}{hidden Markov model}
\newac{dft}{DFT}{discrete time Fourier transform}
\newac{pnp}{PnP}{Plug-and-Play}
\newac{ode}{ODE}{ordinary differential equation}
\newac{gps}{GPs}{Gaussian processes}
\newac{ai}{AI}{artificial intelligence}
\newac{2d}{2D}{two-dimensional}
\newac{ddpm}{DDPM}{denoising diffusion probabilistic modeling}

\setkeys{Gin}{width=1.0\columnwidth}

\providecommand{\propositionname}{Proposition}
\providecommand{\remarkname}{Remark}

\newac{uaps}{UAPS}{utility-aware path search}
\newac{hqs}{HQS}{half quadratic splitting}
\newac{psnr}{PSNR}{peak signal-to-noise ratio}

\newac{tsp}{TSP}{traveling salesperson problem}
\newac{tfmpnp}{TFM-PnP}{truncated flow matching plug-and-play}
\makeatother

\providecommand{\propositionname}{Proposition}

\begin{document}
	\title{Active Learning for Low-Altitude Radio Map Construction via Plug-and-Play
		Flow Matching}
	\maketitle
	\begin{abstract}
		The deployment of unmanned aerial vehicles (UAVs) in low-altitude
		airspace requires accurate and timely radio maps for reliable
		communication and safe navigation. However, constructing such radio maps
		is challenging due to the prohibitive overhead of exhaustive measurements
		and the limited flight endurance of UAVs. To address this challenge, we propose an active learning framework
		based on flow matching for efficient low-altitude radio map construction
		from sparse measurements. We first analyze a plug-and-play (PnP) inference scheme with a flow-matching prior. By characterizing the late-stage refinement behavior through an ordinary differential equation (ODE), we theoretically show how the inference steps smoothly align with a continuous ODE flow to refine the map details. Recognizing that the early generative stages are largely noise-dominated, this insight motivates our proposed truncated flow matching plug-and-play (TFM-PnP) approach. TFM-PnP utilizes a spatial interpolation-based initialization to start the reconstruction from an intermediate flow time, thereby bypassing the inefficient early stages. We further use the generative diversity of flow matching to derive an uncertainty map to guide the UAV trajectory design. Specifically, we propose a weighted sampling approach to select a target location, followed by a Utility-Aware Path Search (UAPS) algorithm to design the corresponding UAV trajectories. Simulation results based on Sionna ray-tracing datasets
		show that the proposed framework outperforms the considered
		baselines, achieving more than $50$\% reduction in normalized mean squared error
		(NMSE).
	\end{abstract}

	\section{Introduction}

	The rapid proliferation of \ac{uavs} for autonomous delivery \cite{HosNasTar:J16},
	aerial monitoring \cite{WanHuaShaGao:J25}, and wireless service provisioning
	\cite{WuMWuHLuW:J25} is transforming low-altitude operations from
	isolated, mission-specific flights into coordinated, infrastructure-supported
	activities across both urban and rural areas. However, the low-altitude
	airspace is characterized by a dynamic and heterogeneous radio
	environment, where obstacles such as buildings, vegetation, and moving
	objects lead to severe scattering and blockage. As a result, the wireless
	channel conditions for UAV links can vary sharply even over short
	distances, posing challenges for maintaining reliable and
	continuous connectivity.

	To address this issue, radio maps have emerged as an effective means
	for characterizing and predicting spatial variations in channel quality
	\cite{SunChe:J24b,XuLCheChePuW:J26,SunChe:J24}. A radio map offers
	fine-grained spatial information about channel characteristics, such
	as the \ac{rss}, effectively serving as a digital twin of the radio
	environment \cite{TimShrFux:J24,ZenCheXuWu:J24,SunChe:J22,LiuYuGao:J25,ShrFuHong:J22}.
	Such knowledge is essential for maintaining robust command-and-control
	links, optimizing spectrum utilization, and ensuring safe navigation
	and trajectory planning, particularly in complex areas where signal
	propagation is highly irregular \cite{MoxHuaXuj:J19,CheLiBSunCui:J25,LiBChe:J24}.
	Nevertheless, constructing accurate low-altitude radio maps remains
	challenging. The need to cover wide areas with UAVs of limited
	battery life makes dense measurements impractical in both time and
	cost. Therefore, developing intelligent sampling and reconstruction
	methods that can generate accurate radio maps from a limited number
	of strategically collected measurements remains a critical problem.

Existing approaches for intelligent sampling and reconstruction can 
be categorized into two groups. Traditional methods leverage statistical 
models such as \ac{gps} to interpolate sparse data. Although GPs can quantify 
predictive uncertainty and thereby guide subsequent sampling 
\cite{CheZhuWanLin:C25,PolSadYeW:J24}, they do not scale well and often 
struggle in complex blockage-dominated scenarios. More recent studies 
have used deep learning techniques \cite{CheWanGuo:J25,LuWGaoWenLia:C25,ShrRomChe:J23}, 
including autoencoders and U-Nets, to learn latent structures of the 
map from data. Unlike GPs, these models do not inherently provide 
calibrated uncertainty estimates. To address this issue, researchers 
have introduced auxiliary uncertainty-estimation methods, such as Bayesian 
neural networks or deep ensembles \cite{LuWGaoWenLia:C25,ShrRomChe:J23}. 
However, these methods are computationally expensive and can fail to 
yield reliable uncertainty estimates because they do not model prediction 
uncertainty explicitly. Furthermore, existing active sensing methods typically 
update the uncertainty estimate after each newly acquired measurement and then 
determine the next sensing action according to a local uncertainty or acquisition 
criterion, sometimes with additional distance-based heuristics. For example, 
\cite{PolSadYeW:J24} queried one new location at a time via acquisition functions 
with an optional distance penalty, whereas \cite{ShrRomChe:J23} replanned 
toward a single high-uncertainty destination using an uncertainty-aware route. 
These methods either do not explicitly incorporate flight distance into target 
location selection or require uncertainty re-estimation and replanning after 
acquiring each new sensing location.

	In recent years, the emergence of generative \ac{ai} has provided
	new tools for modeling complex data distributions. Early
	progress has been achieved by diffusion models \cite{HoJJaiAbb:J21,sohldickstein2015deepunsupervisedlearningusing,song2021scorebasedgenerativemodelingstochastic},
	which enable restoration through iterative inference with gradient-based
	measurement guidance \cite{mardani2023variationalperspectivesolvinginverse,song2022solvinginverseproblemsmedical}.
	Although diffusion models are effective, their inference typically requires
	tens to hundreds of discrete denoising steps, each involving gradient
	evaluations or forward passes through the score network, leading
	to substantial computational overhead. In contrast, flow matching
	\cite{LipCheBenNic:J23,liu2022flowstraightfastlearning} learns a
	continuous-time velocity field that transports a simple base distribution
	to the data distribution. Although sampling still involves numerically
	integrating this velocity field, it avoids the multi-level noise scheduling
	and iterative denoising structure of diffusion models, resulting in
	a more direct and often more efficient generative process.

	Recent studies on flow matching-based restoration have explored various
	strategies for incorporating sparse measurement information into the
	learned dynamics to reconstruct complete images. Some works \cite{kim2025flowdpsflowdrivenposteriorsampling,pokle2024trainingfreelinearimageinverses}
	injected the measurement likelihood directly into the inference stage
	by adjusting the velocity field, enabling efficient posterior sampling
	without retraining. Other works \cite{benhamu2024dflowdifferentiatingflowscontrolled,zhang2025flowpriorslinearinverse}
	coupled the measurement-consistency constraint explicitly with the
	flow trajectory through solver differentiation or time-dependent \ac{map}
	optimization, achieving high reconstruction accuracy at the cost of
	increased numerical sensitivity. Alternatively, the flow matching
	prior can be incorporated into iterative \ac{pnp} schemes \cite{MarGagHag:J25},
	where it is implemented as a time-dependent flow-induced prior operator that alternates with measurement
	alignment and constraint projection. However, many of these designs integrate measurement guidance throughout the full generative trajectory or rely on tightly coupled time-dependent optimization during inference. Under sparse observations, such whole-trajectory correction can be inefficient and numerically delicate, particularly in the noise-dominated early stage.
	Moreover, existing flow-based approaches have primarily focused on generation and reconstruction. Their potential for explicit uncertainty characterization remains underexplored, especially in UAV-assisted radio map construction with trajectory-aware active sensing.
	\begin{figure*}
		\includegraphics[width=1\textwidth]{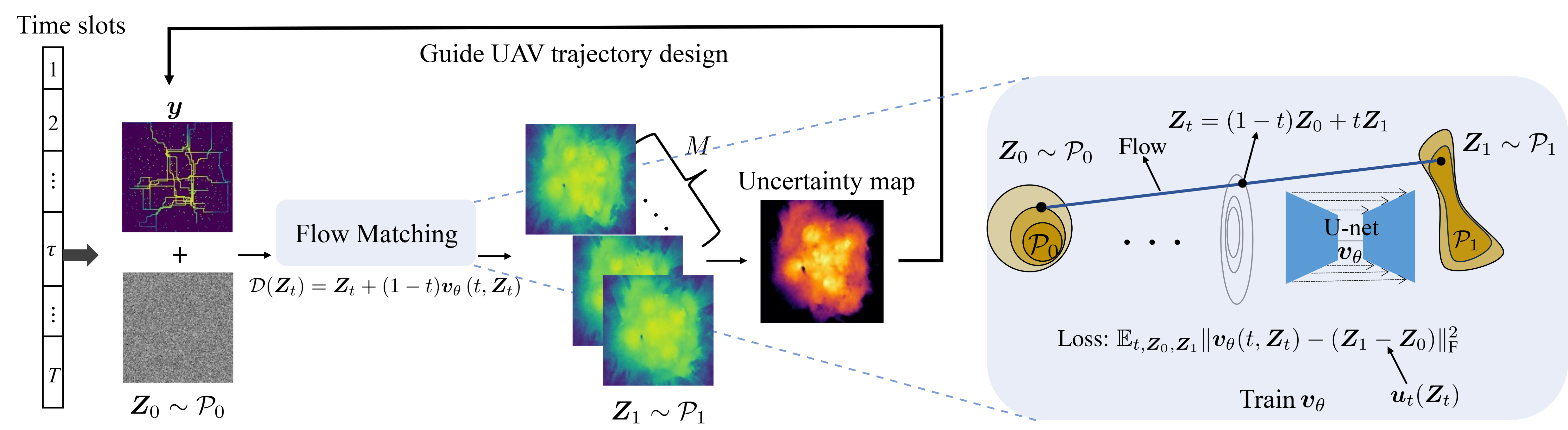}\caption{\label{fig:Active-learning-process}Flow matching-based active learning
			process for radio map construction.}
	\end{figure*}

	To address the above issues, this paper develops an uncertainty-aware
	active learning framework for UAV-assisted low-altitude radio map
	construction from sparse UAV measurements. Our goal is threefold:
	1) to develop a stable and efficient flow matching-based reconstruction method; 2) to quantify the uncertainty of the reconstructed radio map; and 3) to use this uncertainty for sequential UAV trajectory design under a limited sensing budget. 
	To this end, we develop a flow-based active learning framework
	(cf. Fig.~\ref{fig:Active-learning-process}) that tightly couples
	stable reconstruction with UAV trajectory planning through an uncertainty
	map derived from the reconstructed radio maps. Specifically, the framework
	first incorporates a flow matching prior into a \ac{pnp} reconstruction
	scheme to produce measurement-consistent radio maps. It then derives
	pixel-wise uncertainty using the generative variability of
	flow matching. Finally, it uses the resulting uncertainty map to plan
	sequential UAV trajectories that maximize the uncertainty reward under a limited flight
	budget.

	The main contributions of this work are summarized as follows:
	\begin{itemize}
		\item We theoretically analyze a \ac{pnp} inference scheme with a flow-matching prior. By modeling the late-stage refinement behavior through an \ac{ode}, we reveal the underlying refinement dynamics. Specifically, we show how the discrete inference steps smoothly align with a continuous \ac{ode} flow to recover map details. Motivated by this insight, we propose the \ac{tfmpnp} approach. By utilizing a spatial interpolation-based initialization, \ac{tfmpnp} starts the reconstruction from an intermediate flow time and performs late-stage flow-guided refinement toward \(t=1\).
		\item We develop a sample-based uncertainty estimation scheme that connects \ac{tfmpnp} reconstruction with active UAV trajectory planning. By generating multiple radio map reconstructions through repeated runs of \ac{tfmpnp}, the proposed method estimates spatial uncertainty via pixel-wise variance. The resulting uncertainty map serves as the basis for uncertainty-guided sensing.
		\item We further design a UAV trajectory-planning strategy that accounts for both uncertainty reward and flight cost. In particular, at each replanning round, we propose a weighted sampling approach to select a target by maximizing the uncertainty reward under a limited flight budget. Subsequently, \ac{uaps} generates a trajectory from the UAV's current location to the target location.
		\item Experimental results on Sionna ray-tracing datasets show that the proposed \ac{tfmpnp}-based active learning framework outperforms representative baselines, yielding more than $50\%$ \ac{nmse} reduction.
	\end{itemize}

	The remainder of the paper is organized as follows. Section~\ref{sec:Problem-Statement} formulates the UAV-assisted radio map reconstruction problem. Section~\ref{sec:Preliminaries} reviews the PnP and flow matching preliminaries. Section~\ref{sec:Time-Truncated-Gradient}  analyzes the asymptotic ODE limit of the PnP-Flow iteration and presents the proposed TFM-PnP reconstruction algorithm. Section~\ref{sec:Active-Learning-via} develops the uncertainty-aware active learning module. Section~\ref{sec:Simulation-Results} presents the simulation results, and Section~\ref{sec:Conclusion} concludes the paper.

	\emph{Notation:} Scalars are denoted by italic letters (e.g., $x$). 
	Bold lowercase letters denote vectors or grid locations (e.g., $\bm{x}$ and $\bm{g}$), 
	and bold uppercase letters denote matrices or grid-valued maps (e.g., $\bm{X}$). 
	The entry at the $i$-th row and $j$-th column of matrix $\bm{X}$ is denoted by $\bm{X}_{[i,j]}$. 
	Calligraphic letters are used for sets, operators, distributions, or functionals, depending on the context; such objects are explicitly defined when introduced. 
	$\mathbb{R}$ denotes the set of real numbers, $\bm I$ denotes an identity matrix with compatible dimension, and $\mathbb{E}[\cdot]$ denotes expectation. 
	$\mathcal{N}(\bm{\mu},\bm{\Sigma})$ denotes a Gaussian distribution with mean $\bm{\mu}$ and covariance $\bm{\Sigma}$. 
	The notation $U[a,b]$ denotes the continuous uniform distribution on the interval $[a,b]$, and $|\mathcal A|$ denotes the cardinality of set $\mathcal A$. 
	Unless otherwise specified, $\|\cdot\|$ denotes the Euclidean norm for vectors and the Frobenius norm for matrices; $\|\cdot\|_{\mathrm F}$ is used when the matrix norm needs to be emphasized.

	\section{Problem Statement\label{sec:Problem-Statement}}

	\subsection{Measurement Model}

	Consider a low-altitude urban environment where multiple \ac{bss}
	are deployed, either at ground level or on building rooftops. The
	objective is to estimate the spatial distribution of the aggregated \ac{rss}
	over a low-altitude \ac{2d} horizontal area \(\mathcal{S}\subset\mathbb{R}^{2}\)
	from sparse RSS measurements collected by a UAV. This area is discretized into a uniform grid of size
	$I\times J$, and the RSS value at each grid index $(i,j)$ is
	represented by the entry \(\bm H_{[i,j]}\), where \(\bm H \in \mathbb R^{I\times J}\) is the unknown ground-truth radio map.

	A UAV collects RSS measurements while flying over the area. Let \(\Omega^{(\tau)} \subseteq \{1,\dots,I\}\times\{1,\dots,J\}\) denote the set of grid cells visited by the UAV up to time slot \(\tau\), and let \(Q_\tau = |\Omega^{(\tau)}|\) denote the number of sampled locations. For each sampled cell $(i_q, j_q) \in \Omega^{(\tau)}$, the measurement is modeled as $y_q^{(\tau)} = \bm{H}_{[i_q, j_q]} + n_q$, where $n_q \sim \mathcal{N}(0, \sigma_y^2)$ denotes the measurement noise associated with the $q$-th sample.

	Stacking all measurements gives $\bm{y}^{(\tau)} = \mathcal{H}_{\Omega^{(\tau)}}(\bm{H}) + \bm{n}^{(\tau)}$, where $\bm{n}^{(\tau)} \sim \mathcal{N}(\bm{0}, \sigma_y^2 \bm{I})$ and $\mathcal{H}_{\Omega^{(\tau)}}(\cdot) : \mathbb{R}^{I \times J} \to \mathbb{R}^{Q_\tau}$ extracts the entries indexed by $\Omega^{(\tau)}$. The adjoint operator 
	$\mathcal H_{\Omega^{(\tau)}}^{*}(\cdot): \mathbb R^{Q_\tau}\to \mathbb R^{I\times J}$
	maps a measurement-domain vector back to the grid domain by placing its
	entries at the sampled locations in $\Omega^{(\tau)}$ and filling zeros
	elsewhere.

	The objective is to recover the complete radio map \(\bm H\) from sparse measurements \(\bm y^{(\tau)}\) and their sampling locations \(\Omega^{(\tau)}\).

	\subsection{Active Learning Framework}

	The UAV-assisted radio map reconstruction is performed sequentially over discrete time slots \(\tau=1,\dots,T\). At each time slot \(\tau\), all measurements collected up to time \(\tau\), denoted by \(\bm y^{(\tau)}\), are used to generate a set of radio map reconstruction samples \(\{\bm Z^{(\tau,m)}\}_{m=1}^M\). Based on these samples, an uncertainty map is computed to quantify the confidence of the estimated RSS values over the entire grid. The resulting uncertainty map is then used to select an informative target and plan the UAV trajectory for the next time slot \(\tau+1\), so that future measurements are preferentially acquired from informative regions. This closed-loop process of reconstruction, uncertainty quantification, and trajectory planning is repeated until the flight budget is exhausted.

	As illustrated in Fig.~\ref{fig:Active-learning-process}, the proposed framework tightly couples data-driven radio map reconstruction with uncertainty-aware measurement acquisition, enabling efficient and accurate radio map construction from sparse UAV observations.

	\section{Preliminaries\label{sec:Preliminaries}}

	This section briefly reviews the standard \ac{pnp} framework for solving inverse problems, along with the theoretical foundations of flow matching generative models. We present both in the context of radio map construction. Our primary focus lies in the generative inference of a single reconstructed radio map $\bm{Z}^{(\tau,m)}$, at a given time slot $\tau$. For notational simplicity, the dependence on both the time slot index $\tau$ and the sample index $m$ is omitted.
\subsection{Radio Map Construction via a PnP Framework}

We first introduce the standard PnP framework for radio map construction.
This inverse problem can be formulated as a \ac{map} estimation problem
\begin{equation}
	\max_{\bm{Z}}\ \left\{ \log f(\bm{y}|\bm{Z})+\log p(\bm{Z})\right\},
	\label{eq:map_formulation}
\end{equation}
where \(f(\bm{y}|\bm{Z})\) denotes the likelihood of the measurements
\(\bm{y}\) given the radio map \(\bm{Z}\), and \(p(\bm{Z})\) is the prior
distribution of radio maps.

Equivalently, \eqref{eq:map_formulation} can be reformulated as the following
regularized inverse problem:
\begin{equation}
	\min_{\bm{Z}}\ \left\{ F(\bm{Z};\bm{y})+R(\bm{Z})\right\},
	\label{eq:regularized_problem}
\end{equation}
where
\(F(\bm{Z};\bm{y})\coloneqq-\log f(\bm{y}|\bm{Z})\) is the measurement
consistency term, and \(R(\bm{Z})\coloneqq-\log p(\bm{Z})\) encodes the
structural prior of radio maps. 	A common choice for the measurement consistency term is $F(\bm{Z};\bm{y})=\frac{1}{2}\|\mathcal{H}_\Omega(\bm{Z})-\bm{y}\|^{2}$.

Given the minimization problem in \eqref{eq:regularized_problem}, the proximal
gradient method \cite{BoyParChu:B11} solves it by alternating between a measurement-consistency step
and a prior-consistency step:
\begin{equation}
	\left\{
	\begin{aligned}
		\bar{\bm{Z}}^{(k+1)}
		&=
		\bm{Z}^{(k)}
		-
		\eta_k
		\nabla_{\bm{Z}} F(\bm{Z}^{(k)};\bm{y}), \\
		\bm{Z}^{(k+1)}
		&=
		\operatorname{prox}_{\eta_k R}
		\left(
		\bar{\bm{Z}}^{(k+1)}
		\right),
	\end{aligned}
	\right.
	\label{eq:proximal_gradient}
\end{equation}
where \(\eta_k\) is the step size. The proximal operator is defined as
\begin{equation}
	\operatorname{prox}_{\eta_k R}(\bar{\bm{Z}})
	=
	\arg\min_{\bm{Z}}
	R(\bm{Z})
	+
	\frac{1}{2\eta_k}
	\left\|
	\bm{Z}-\bar{\bm{Z}}
	\right\|_F^2 .
	\label{eq:proximal_operator}
\end{equation}

However, the explicit form of the radio-map prior \(R(\bm{Z})\) is generally
unknown and difficult to model analytically. Following the plug-and-play
principle, we use a learned prior operator in place of the explicit proximal map in \eqref{eq:proximal_gradient}:
\begin{equation}
	\bm{Z}^{(k+1)}
	=
	\mathcal{D}_{\theta,k}
	\left(
	\bar{\bm{Z}}^{(k+1)}
	\right),
	\label{eq:pnp_operator}
\end{equation}
where \(\mathcal{D}_{\theta,k}(\cdot)\) serves as an implicit prior model for
radio maps.

In this work, we instantiate the learned prior operator $\mathcal{D}_{\theta,k}$ using a generative prior trained via flow matching, as specified in the following subsection.

	\subsection{Flow Matching}

	Flow matching \cite{LipCheBenNic:J23} is an emerging paradigm for
	generative modeling that formulates sample generation as learning
	a continuous-time probability flow governed by an \ac{ode}. Unlike
	conventional diffusion-based generative models, flow matching avoids
	costly simulation during training by reformulating the objective into
	a scalable regression problem.

	\subsubsection{Probability Flows and Velocity Fields}

	Let $\mathcal{P}_{0}$ denote a simple prior distribution, typically
	a standard Gaussian distribution \cite{LipCheBenNic:J23}, and let $\mathcal{P}_{1}$
	denote the target distribution of radio maps over the space $\mathbb{R}^{I\times J}$.
	Note that in \eqref{eq:map_formulation}, $p(\bm{Z})$ refers to the
	probability density function associated with $\mathcal{P}_{1}$.

	Continuous normalizing flows construct a probability path $\{\mathcal{P}_{t}\}_{t\in[0,1]}$
	that continuously interpolates between these two distributions, transforming
	samples from the prior $\mathcal{P}_{0}$ at time $t=0$ to the target
	distribution $\mathcal{P}_{1}$ at time $t=1$. This probability path
	is generated by a corresponding time-dependent velocity field $\bm{u}_{t}(\bm{Z}_{t})$
	defined as:
	\begin{equation}
		\bm{u}_{t}:[0,1]\times\mathbb{R}^{I\times J}\to\mathbb{R}^{I\times J}, \quad (t,\bm{Z}_{t}) \to \bm{u}_{t}(\bm{Z}_{t}),
	\end{equation}
	where each time $t$ corresponds to a distinct continuous vector field.

	The evolution of samples along this probability path is governed by
	an \ac{ode}:
	\begin{equation}
		\frac{\mathrm{d}}{\mathrm{d}t}\psi_{t}(\bm{Z}_{0})=\bm{u}_{t}(\psi_{t}(\bm{Z}_{0})),
	\end{equation}
	with the initial condition $\psi_{0}(\bm{Z}_{0})=\bm{Z}_{0}$ for any $\bm{Z}_{0}\sim\mathcal{P}_{0}$.
	Solving this \ac{ode} yields a deterministic flow map $\psi_{t}$ that
	transports samples along the trajectories defined by the velocity
	field. Consequently, flowing a noise sample $\bm{Z}_{0}$
	via this map yields the latent state at time $t$:
	\begin{equation}
		\bm{Z}_{t}\coloneqq\psi_{t}(\bm{Z}_{0})\sim\mathcal{P}_{t}.
	\end{equation}
	By integrating the velocity field from $t=0$ to $t=1$, the simple
	prior $\mathcal{P}_{0}$ is smoothly and deterministically transported
	to the complex radio map distribution $\mathcal{P}_{1}$.

	\subsubsection{Learning the Velocity Field}

	In practice, the true velocity field $\bm{u}_{t}(\bm{Z}_{t})$ is analytically intractable and must be approximated by a neural network $\bm{v}_{\theta}(t,\bm{Z}_{t})$. To make training scalable, the flow matching framework employs conditional probability paths. A widely used choice is the conditional Optimal Transport (OT) path \cite{LipCheBenNic:J23}, formulated as 
	\begin{equation}
		\bm{Z}_{t}=(1-t)\bm{Z}_{0}+t\bm{Z}_{1} \label{eq:ot_path}
	\end{equation}
	with a constant conditional velocity 
		\begin{equation}
	\bm{u}_{t}(\bm{Z}_{t}\mid\bm{Z}_{1}, \bm{Z}_{0})=\bm{Z}_{1}-\bm{Z}_{0},
		\end{equation}
	given $\bm{Z}_{0}\sim\mathcal{P}_{0}$ and $\bm{Z}_{1}\sim\mathcal{P}_{1}$. A simple illustration of the training process is shown on
	the right side of Fig.~\ref{fig:Active-learning-process}.

	The network $\bm{v}_{\theta}(t,\bm{Z}_{t})$, typically based on a U-Net architecture, is trained by minimizing the Conditional Flow Matching (CFM) loss:
	\begin{align}
		\mathcal{L}_{\text{CFM}}(\theta)=\mathbb{E}_{t,\bm{Z}_{0},\bm{Z}_{1}}\left\Vert \bm{v}_{\theta}\!\big(t,(1-t)\bm{Z}_{0}+t\bm{Z}_{1}\big)-(\bm{Z}_{1}-\bm{Z}_{0})\right\Vert ^{2},\label{eq:cfm_loss}
	\end{align}
	where $t\sim U[0,1]$. Minimizing this surrogate objective has been shown to recover the desired marginal velocity field $\bm{u}_{t}(\bm{Z}_{t})$ \cite{LipCheBenNic:J23}. Once the network is trained, new radio map samples are generated by numerically simulating the ODE $\frac{\mathrm{d}\bm{Z}_{t}}{\mathrm{d}t}=\bm{v}_{\theta}(t,\bm{Z}_{t})$ from a noise prior.
\subsection{PnP using Flow Matching as a Prior}

In the flow-matching implementation, this operator is indexed by the flow
time, i.e., \(\mathcal{D}_{\theta,k}=\mathcal D_{\theta,t_k}\). For
the affine conditional OT path in \eqref{eq:ot_path}, it is specified as
\[
\mathcal D_{\theta,t}(\bm W)
=
\bm W+(1-t)\bm v_{\theta}(t,\bm W).
\]
This expression gives an endpoint estimate at \(t=1\) from a state
\(\bm W\) at flow time \(t\). Thus, \(\mathcal D_{\theta,t}\) is a
time-indexed learned prior update used in the PnP iteration, rather than
the exact proximal map of \(R(\bm Z)\) or a complete numerical integrator
of the flow ODE.

By integrating the measurement-consistency update with the flow-matching prior, the resulting PnP-Flow \cite{MarGagHag:J25} iteration for radio map construction is formulated as:
\begin{equation}
	\left\{
	\begin{aligned}
		\bar{\bm{Z}}^{(k+1)}
		&=
		\bm{Z}^{(k)}
		-
		\eta_k
		\nabla_{\bm{Z}}F(\bm{Z}^{(k)};\bm{y}), \\
		\bm{W}^{(k+1)}
		&=
		t_k\bar{\bm{Z}}^{(k+1)}
		+
		(1-t_k)\bm{Z}_0^{(k)},  \\
		\bm{Z}^{(k+1)}
		&=
		\mathcal{D}_{\theta,t_k}
		\left(
		\bm{W}^{(k+1)}
		\right),
	\end{aligned}
	\right.
	\label{eq:pnp_flow_iteration}
\end{equation}
where \(t_k\in[0,1)\) is a monotonically increasing time schedule corresponding to the flow evolution, with \(t_0=0\) and \(\lim_{k\rightarrow\infty}t_k=1\).

The three-step iteration in \eqref{eq:pnp_flow_iteration} has a direct interpretation. The first step takes a gradient step to pull the current estimate toward measurement consistency. The second step reparameterizes the intermediate estimate into the affine form of the conditional optimal transport path defined in \eqref{eq:ot_path}. Finally, the third step applies \(\mathcal D_{\theta,t_k}(\cdot)\) to produce an endpoint estimate along the learned flow path. In this sense, the learned velocity field provides an implicit prior-refinement mechanism without requiring an explicit analytic regularizer \(R(\bm Z)\).

\section{The Proposed TFM-PnP Approach for Radio Map Reconstruction\label{sec:Time-Truncated-Gradient}}

In this section, we present the TFM-PnP reconstruction algorithm for sparse UAV measurements. We begin with an idealized ODE interpretation of the late-stage PnP-Flow iteration, corresponding to \(t_k\to1\), which identifies the late-stage dynamics as a measurement-consistent refinement process under a stochastic-approximation scaling. This interpretation motivates the use of a truncated flow trajectory initialized from a coarse spatial prior. We then construct the time-truncated initialization and summarize the resulting finite-step inference procedure.

\subsection{Asymptotic ODE Limit of PnP-Flow}

We first examine the late-stage PnP-Flow recursion in
\eqref{eq:pnp_flow_iteration}, where \(t_k\to1\). In this regime, the factor \(h_k=1-t_k\) becomes the small coefficient associated with the endpoint prediction. Under a stochastic-approximation scaling, the resulting late-stage recursion admits the following ODE-method characterization.

\begin{prop}[Asymptotic ODE Limit of PnP-Flow]
	\label{prop:ode_limit}
	Fix a finite index \(k_0\) and consider the sequence of the
	PnP-Flow recursion in \eqref{eq:pnp_flow_iteration} for \(k\ge k_0\).
	Let \(h_k\triangleq 1-t_k\) and
	\(\beta_k\triangleq t_k\eta_k/h_k\), where \(\eta_k\) denotes the
	gradient step size. Assume that: (i)
	\(\sum_{k=k_0}^{\infty} h_k=\infty\)
	and \(\sum_{k=k_0}^{\infty} h_k^2<\infty\); (ii)
	\(F(\bm Z; \bm y)=\frac{1}{2}\|\mathcal{H}_{\Omega}(\bm Z)-\bm y\|^2\),
	\(F\in C^1\), and \(\nabla F\) is locally Lipschitz; (iii)
	\(\bm v_{\theta}(t,\bm Z)\) is continuous in \(t\in[0,1]\) and
	globally Lipschitz in \(\bm Z\) uniformly in \(t\); (iv)
	\(\beta_k\to\beta\) for some \(\beta\in[0,\infty)\); (v) the
	reference samples \(\{\bm Z_0^{(k)}\}\) are independent of the past,
	have zero mean, and have uniformly bounded second moments; and (vi)
	the iterate sequence \(\{\bm Z^{(k)}\}_{k\ge k_0}\) is almost surely bounded.

	Let \(s_{k_0}=0\), \(s_{k+1}=s_k+h_k\) for \(k\ge k_0\), and let
	\(\widetilde{\bm Z}(s)\) be the piecewise linear interpolation of
	the discrete iterates over the intervals \([s_k,s_{k+1})\). Then
	\(\widetilde{\bm Z}(s)\) is an asymptotic pseudo-trajectory of the
	autonomous ODE
	\begin{equation}
		\frac{d\bm Z(s)}{ds}
		=
		-\bm Z(s)
		-
		\beta\nabla F(\bm Z(s))
		+
		\bm v_{\theta}(1,\bm Z(s)).
		\label{eq:asymptotic_ode}
	\end{equation}
\end{prop}

\textit{Proof.} The detailed proof is provided in Appendix~\ref{sec:Proof-of-Proposition_ode}.

Proposition~\ref{prop:ode_limit} gives a continuous-time interpretation of
the late-stage PnP-Flow dynamics. The limiting drift in
\eqref{eq:asymptotic_ode} contains three terms: the contraction term
\(-\bm Z(s)\) induced by the affine reparameterization, the
measurement-consistency term \(-\beta\nabla F(\bm Z(s))\), and the learned
velocity field at \(t=1\), \(\bm v_{\theta}(1,\bm Z(s))\). Thus, near the
end of the flow trajectory, the discrete refinement behaves like a
deterministic dynamical system that balances agreement with the observed RSS
samples and prior information from the learned radio-map distribution.

The connection between this late-stage asymptotic model and the practical time-truncated implementation is clarified in the following remark.

\begin{remark}[Connection to time truncation]
	Proposition~\ref{prop:ode_limit} should be interpreted as an
	algorithmic justification rather than a finite-step convergence guarantee
	for Algorithm~\ref{alg:trace_pnp}. It shows that, in the late stage
	where \(t_k\to1\), the PnP-Flow recursion behaves as a continuous refinement
	dynamics that jointly incorporates measurement consistency and the learned
	radio-map prior. This observation suggests that the early part of the
	flow trajectory is not necessarily indispensable for reconstruction if
	an appropriate intermediate state is available.
	
	In particular, the affine OT path in \eqref{eq:ot_path} implies that
	early-time states contain a dominant contribution from the base noise,
	whereas late-stage states closer to \(t=1\) are expected to carry more radio-map structure.
	Therefore, if one can construct a coarse but structured estimate that is
	compatible with an intermediate flow state, the reconstruction can be
	initialized directly at \(t=a\) and then evolved through the late-stage
	refinement steps toward \(t=1\). In this sense, the proposed
	time truncation replaces the noise-dominated early evolution with a
	spatially informed initialization, while retaining the
	late-stage ODE-like refinement mechanism identified in
	Proposition~\ref{prop:ode_limit}.
\end{remark}

\subsection{Time-Truncated Initialization via Spatial Interpolation Prior}

A truncated trajectory requires an initial state at time \(t=a\) that already contains meaningful spatial information. We construct this state from the available sparse measurements through a lightweight spatial interpolation prior.

We exploit the spatial smoothness of radio maps for this purpose. Given the observed RSS values $\bm y$ and their spatial locations specified by the sampling operator $\mathcal{H}_{\Omega}$, we apply a \ac{knn} interpolation scheme to obtain a coarse radio map estimate, denoted by $\bm X_{\mathrm{prior}}$. This interpolation prior captures the large-scale spatial trend of the propagation field while remaining computationally lightweight.

To make the initialization compatible with the affine parameterization of the conditional OT path in \eqref{eq:ot_path}, we combine the interpolation prior with a base noise sample $\bm Z_0^{(0)}\sim\mathcal P_0$ as
\begin{equation}
	\bm Z_a^{(0)}
	=
	a \bm X_{\mathrm{prior}}
	+
	(1-a)\bm Z_0^{(0)} .
	\label{eq:truncated_initialization}
\end{equation}

The initialization in \eqref{eq:truncated_initialization} uses the affine structure of the conditional flow path. Although it is not intended to be an exact point on the true conditional flow trajectory, it provides a practical surrogate state at time \(t=a\) that preserves coarse spatial structure from \(\bm X_{\mathrm{prior}}\) while retaining the affine form used by the flow model. This design replaces the early noise-dominated integration segment with a lightweight spatial initialization, allowing the subsequent late-stage PnP-Flow refinement steps to focus on measurement-consistent spatial refinement.

\subsection{The Proposed TFM-PnP Algorithm}

Combining the time-truncated initialization in \eqref{eq:truncated_initialization} with the PnP-Flow recursion in \eqref{eq:pnp_flow_iteration} yields the finite-step TFM-PnP inference procedure summarized in Algorithm~\ref{alg:trace_pnp}. Starting from \(t=a\), each iteration first enforces measurement consistency through a gradient step, then reparameterizes the corrected state along the affine flow path, and finally applies the flow-induced prior operator for late-stage refinement toward \(t=1\).

In this way, the algorithm implements the preceding insight in a practical finite-step form: early noise-dominated simulation is replaced by spatial interpolation, while the remaining iterations perform late-stage refinement along the flow path toward \(t=1\).

\begin{algorithm}[t]
	\caption{TFM-PnP Framework for Radio Map Reconstruction}
	\label{alg:trace_pnp}
	\begin{algorithmic}[1]
		\State \textbf{Input:} Observed RSS $\bm{y}$; measurement operator $\mathcal{H}_{\Omega}$; trained velocity field $\bm{v}_{\theta}$; spatial empirical prior $\bm{X}_{\text{prior}}$; base noises $\{\bm{Z}_0^{(k)}\}_{k=0}^{K-1}$ with $\bm{Z}_0^{(k)}\sim\mathcal{P}_0$; gradient step sizes $\{\eta_k\}_{k=0}^{K-1}$; number of inference steps $K\ge 2$; start time $a\in(0,1)$
		\State \textbf{Initialize state:} $\bm{Z}^{(0)}=a\bm{X}_{\text{prior}}+(1-a)\bm{Z}_0^{(0)}$
		\For{$k=0,\dots,K-1$}
		\State $t_k = a+(1-a)\frac{k}{K}$
		
		\State \textit{\# Step 1: Measurement Consistency}
		\State Compute
		\[
		\bar{\bm{Z}}^{(k+1)}
		=
		\bm{Z}^{(k)}
		-
		\eta_k
		\mathcal{H}_{\Omega}^{*}
		\big(
		\mathcal{H}_{\Omega}(\bm{Z}^{(k)})-\bm{y}
		\big)
		\]
		
		\State \textit{\# Step 2: Flow-Inspired Reparameterization}
		\State Compute
		\[
		\bm{W}^{(k+1)}
		=
		t_k\bar{\bm{Z}}^{(k+1)}
		+
		(1-t_k)\bm{Z}_0^{(k)}
		\]
		
		\State \textit{\# Step 3: Flow-Induced Prior Refinement}
		\State Compute
		\[
		\bm{Z}^{(k+1)}
		=
		\bm{W}^{(k+1)}
		+
		(1-t_k)\bm{v}_{\theta}
		\left(
		t_k,\bm{W}^{(k+1)}
		\right)
		\]
		
		\EndFor
		\State \Return $\bm{Z}^{(K)}$
	\end{algorithmic}
\end{algorithm}

The effect of starting from an intermediate time $t=a$ is examined in Fig.~\ref{fig:fix_timestep} in the simulation results.

	\section{Active Learning via Uncertainty Quantification\label{sec:Active-Learning-via}}

	The proposed TFM-PnP approach enables accurate radio map reconstruction
	within each time slot $\tau$. However, due to the UAVs'
	limited flight budget, it remains challenging to efficiently collect
	informative samples across time. To address this issue, in this section,
	we introduce an active trajectory design framework that maximizes the uncertainty reward under a limited flight budget. Specifically,
	we first develop an uncertainty quantification method based on generative
	samples, and then design a trajectory planning strategy that uses
	the resulting uncertainty map to guide the UAV along the most informative
	flight path.

\subsection{Uncertainty Quantification with Generative Samples}

The generative nature of the proposed \ac{tfmpnp} approach naturally enables sample-based characterization of reconstruction uncertainty. Instead of yielding a single deterministic output, we generate a set of $M$ plausible radio map reconstructions $\{\bm{Z}^{(\tau,m)}\}_{m=1}^{M}$ that are approximately consistent with the measurements $\bm{y}^{(\tau)}$ collected up to time slot $\tau$. This is achieved by executing Algorithm~\ref{alg:trace_pnp} $M$ times, each initialized with an independently drawn Gaussian reference sequence \(\{\bm{Z}_{0}^{(k)}\}_{k=0}^{K-1}\) with \(\bm{Z}_{0}^{(k)}\sim\mathcal{P}_{0}\). To characterize the dispersion of these plausible reconstructions induced by different latent initializations, we define the uncertainty map $\bm{U}^{(\tau)}$ at time slot $\tau$ as the element-wise sample variance:
\begin{equation}
	\bm{U}^{(\tau)}=\frac{1}{M}\sum_{m=1}^{M}\left(\bm{Z}^{(\tau,m)}-\frac{1}{M}\sum_{n=1}^{M}\bm{Z}^{(\tau,n)}\right)^{2}.
\end{equation}
This variance serves as a practical proxy for reconstruction uncertainty. In general, regions with dense measurement coverage tend to exhibit smaller sample variance, whereas poorly observed or geometrically complex regions often yield larger variance.

	\subsection{Active Trajectory Planning}

	In each time slot $\tau$, the UAV plans a flight trajectory based on the current uncertainty map $\bm{U}^{(\tau)}$. Unlike existing sequential active sensing schemes that either set the next destination from a local uncertainty criterion \cite{ShrRomChe:J23} or query one new measurement location from an acquisition function \cite{PolSadYeW:J24} in each time slot, we first select one informative and reachable target using a distance-aware weighted uncertainty score and then generate a utility-aware path toward this target via \ac{uaps}.

	Since the trajectory planning is conducted within a fixed time slot
	$\tau$, we omit the superscript $\tau$ in this subsection for notational
	simplicity.

	\subsubsection{Sequential Target Selection}

	In practice, the UAV should not focus only on uncertain regions
	but also consider the flight distance required to reach them. Although
	a target location may have high uncertainty, it can be far from the
	UAV's current position and therefore be inefficient
	to explore, in the sense that the uncertainty reduction achieved per
	unit flight distance is low. By contrast, nearby regions with slightly
	lower uncertainty may offer better exploration efficiency.

	We define $\mathcal{V}$ as the set of sampled locations
	and $\mathcal{V}_{\text{u}}$ as the set of unsampled locations. To
	account for both uncertainty and flight distance, each
	unsampled grid $\bm{g}\in\mathcal{V}_{\text{u}}$ is assigned a weighted
	uncertainty:
	\begin{equation}
		w\left(\bm{g}\right)=\frac{\rho_{B}\left(\bm{g}\right)}{1+\kappa d\left(\bm{g}_{\text{UAV}},\bm{g}\right)},\label{eq:selection_weight}
	\end{equation}
	where $\bm{g}_{\text{UAV}}$ represents the location of the UAV at
	the beginning of time slot $\tau$, $d(\cdot,\cdot)$ denotes the
	Manhattan distance, and $\kappa\ge0$ controls the trade-off between
	informativeness and reachability. The numerator $\rho_{B}(\bm{g})$
	quantifies the local uncertainty density, defined as
	\begin{equation}
		\rho_{B}\!\left(\bm{g}\right)=\frac{1}{B^{2}}\sum_{\bm{g}'\in\mathcal{B}_{B}(\bm{g})}\xi(\bm{g}'),\label{eq:local_density}
	\end{equation}
	where the uncertainty value is \(\xi(\bm{g}')\triangleq\bm{U}_{[i',j']}\)
	at the grid cell $\bm{g}'=(i',j')$, and $\mathcal{B}_{B}(\bm{g})$
	is a $B\times B$ neighborhood centered at $\bm{g}$. Using a neighborhood
	instead of a single grid value yields a more stable estimate of the
	local uncertainty. Single-cell uncertainty values may fluctuate sharply
	due to noise or prediction artifacts, whereas averaging over a $B\times B$
	region captures the underlying spatial structure. A smaller $B$ captures
	sharp local peaks, whereas a larger $B$ reflects broader uncertain
	areas. The additive constant \(1\) in the denominator avoids division by zero
	when the target is close to the UAV and ensures that the distance
	penalty changes smoothly with distance. In practice, both the local uncertainty
	density $\rho_{B}(\bm{g})$ and the distance $d(\cdot,\cdot)$ are
	min-max normalized to $[0,1]$ before evaluating the weighted uncertainty
	to ensure that the trade-off parameter $\kappa$ remains scale-invariant.

	The weighted uncertainty $w\left(\bm{g}\right)$ thus combines informativeness
	(through $\rho_{B}$) and reachability (through $d$) into a unified
	metric, encouraging the UAV to focus on valuable yet reachable sensing
	locations. We then select a single target location $\bm{g}^{\star}\in\mathcal{V}_{\text{u}}$
	via weighted sampling with probability
	\begin{equation}
		p(\bm{g})=\frac{w(\bm{g})}{\sum_{\tilde{\bm{g}}\in\mathcal{V}_{\text{u}}}w(\tilde{\bm{g}})}.
	\end{equation}
	Compared with deterministically selecting the maximum-weight cell,
	weighted sampling favors high-utility locations while preserving exploration
	diversity across replanning rounds.

	\subsubsection{Trajectory Planning Formulation }

	After the target $\bm{g}^{\star}$ is selected, the next step is to
	design a UAV trajectory $\mathcal{Q}=(\bm{g}_{0}=\bm{g}_{\text{UAV}},\,\bm{g}_{1},\,\ldots,\,\bm{g}_{L}=\bm{g}^{\star})$
	that terminates at $\bm{g}^{\star}$, where $L$ denotes the total
	number of movement steps. The trajectory may contain multiple intermediate
	grid cells between the UAV's current location and the
	selected target.

	For tractability, we assume that the UAV can move only in
	the four cardinal directions with unit step cost. Although this Manhattan-distance-based
	abstraction simplifies the actual UAV kinematics, it provides a sufficiently
	accurate approximation of the vehicle motion when the spatial discretization
	is fine enough. Thus, the consecutive locations in $\mathcal{Q}$
	satisfy $\bm{g}_{l}-\bm{g}_{l-1}\in\{(\pm1,0),(0,\pm1)\}$. Therefore,
	the total flight cost is defined as \(C(\mathcal{Q})=L\), which
	represents the total number of movement steps along the trajectory.
	To evaluate the information utility of a trajectory, we define the
	accumulated uncertainty reward \(R_{\mathrm{unc}}(\cdot)\) as
	\begin{equation}
		R_{\mathrm{unc}}(\mathcal{Q})=\sum_{l=1}^{L}\xi(\bm{g}_{l}),
	\end{equation}
	which measures the total uncertainty collected along the path.

	To achieve efficient exploration, we seek a trajectory that collects
	large uncertainty reward while incurring low flight cost, subject
	to terminating at the selected target. This leads to the following
	ratio-form objective:
	\begin{equation}
		\max_{\mathcal{Q}}  \frac{R_{\mathrm{unc}}(\mathcal{Q})}{C(\mathcal{Q})}\label{eq:ratio_objective}
	\end{equation}

Directly solving \eqref{eq:ratio_objective} generally requires iterative
fractional programming \cite{dinkelbach1967nonlinear}, which can be
computationally expensive for real-time UAV routing on dense grids.
In fractional programming, a ratio-form objective can be handled through
a parametric subtractive form \(R_{\mathrm{unc}}(\mathcal{Q})-\phi C(\mathcal{Q})\),
where the parameter \(\phi\) is iteratively updated until convergence.
Inspired by this subtractive reformulation, we avoid the online update
of \(\phi\) and adopt a fixed-parameter scalarized surrogate:
\begin{equation}
	\max_{\mathcal{Q}}\ R_{\mathrm{unc}}(\mathcal{Q})-(1+\lambda)C(\mathcal{Q}).
	\label{eq:scalarized_objective}
\end{equation}
Here, \(\lambda\ge 0\) controls the trade-off between uncertainty
collection and flight cost.

The scalarized objective in \eqref{eq:scalarized_objective} provides
a step-level interpretation for trajectory generation. Since
\(C(\mathcal{Q})\) represents the number of movement steps, each move
naturally incurs a unit flight cost. Meanwhile, the reward term
\(R_{\mathrm{unc}}(\mathcal{Q})\) indicates that moving through grids with higher
uncertainty is more informative. Therefore, instead of searching for
the shortest path purely in terms of distance, we assign each movement
a utility-aware cost, where the cost of moving into a grid cell is
discounted according to its uncertainty. This
design preserves the basic penalty on path length, while encouraging
the UAV to pass through informative intermediate regions when the
additional uncertainty gain is worth the extra travel distance.

Based on this interpretation, we propose \ac{uaps} to construct a utility-aware trajectory
from the current UAV location \(\bm{g}_{0}=\bm{g}_{\text{UAV}}\) to
the selected target \(\bm{g}_{L}=\bm{g}^{\star}\). It evaluates each
explored grid cell \(\bm{g}\) through a composite evaluation function
\[
\varphi(\bm{g})=r(\bm{g})+h(\bm{g}),
\]
where \(r(\bm{g})\) denotes the accumulated utility-aware path cost
from \(\bm{g}_{0}\) to the current cell \(\bm{g}\), and \(h(\bm{g})\)
estimates the remaining cost from \(\bm{g}\) to the selected target.

	\begin{itemize}
		\item $r(\bm{g})$ represents the flight cost penalized by the collected
		uncertainty along the trajectory from $\bm{g}_{0}$ to
		the current location $\bm{g}$:
		\begin{equation}
			r(\bm{g})=\sum_{k=0}^{m-1}c_{\text{step}}\big(\bm{s}_{k},\bm{s}_{k+1}\big),
		\end{equation}
		where \(\bm{s}_{0}=\bm{g}_{0}\), \(\bm{s}_{m}=\bm{g}\),
		and \((\bm{s}_{0},\ldots,\bm{s}_{m})\) denotes the sequence of \(m+1\)
		grid cells on the local subpath, and $c_{\text{step}}(\bm{s}_{k},\bm{s}_{k+1})$ is the microscopic
		step cost of moving from cell $\bm{s}_{k}$ to $\bm{s}_{k+1}$, defined
		as:
		\begin{equation}
			c_{\text{step}}\big(\bm{s}_{k},\bm{s}_{k+1}\big)=1-\alpha\,\tilde{\xi}\big(\bm{s}_{k+1}\big).
		\end{equation}
		Here, \(\tilde{\xi}(\bm{s}_{k+1})=\frac{\xi(\bm{s}_{k+1})-\min(\bm{U})}{\max(\bm{U})-\min(\bm{U})}\)
		denotes the normalized uncertainty at location $\bm{s}_{k+1}$, $\bm{U}$
		is the global uncertainty map, and $\alpha=1/(1+\lambda)\in(0,1]$.
		\item The heuristic term $h(\bm{g})$ estimates the remaining travel cost from the current cell $\bm{g}$ to the target $\bm{g}_{L}$:
		\begin{equation}
			h(\bm{g})=(1-\alpha)\,d(\bm{g},\bm{g}_{L}),
		\end{equation}
		Since $1-\alpha$ is the minimum possible step cost, the heuristic never overestimates the true surrogate path cost and is therefore admissible.
	\end{itemize}
	At each iteration, the algorithm extracts the cell with the minimum
	composite score $\varphi(\bm{g})$, expands its unvisited neighbors
	to update their accumulated weights $r(\cdot)$ and heuristic costs
	$h(\cdot)$. This search strategy proceeds iteratively until the target
	cell $\bm{g}_{L}$ is extracted, yielding the continuous
	trajectory $\mathcal{Q}$.

	The overall procedures for the sequential active trajectory planning
	are summarized in \textbf{Algorithm~\ref{alg:active_uaps}}.
	\begin{algorithm}[t]
		\caption{Sequential Active Trajectory Planning}
		\label{alg:active_uaps}
		\begin{algorithmic}[1]
			\State \textbf{Input:} Current UAV location \(\bm{g}_{\text{UAV}}\), uncertainty map \(\bm{U}\), unsampled set \(\mathcal{V}_{\text{u}}\), and parameters \(B, \kappa, \lambda\)
			\State \textbf{Output:} Selected target \(\bm{g}^{\star}\) and flight trajectory $\mathcal{Q}$

			\vspace{0.05in}
			\State \textbf{Stage 1: Sequential Target Selection}
			\ForAll{$\bm g\in\mathcal V_{\text{u}}$}
			\State Compute $\rho_B(\bm g)=B^{-2}\sum_{\bm g'\in\mathcal B_B(\bm g)}\xi(\bm g')$ and $d(\bm g_{\text{UAV}},\bm g)$
			\EndFor
			\State Min-max normalize $\rho_B(\bm g)$ and $d(\bm g_{\text{UAV}},\bm g)$ over $\mathcal V_{\text{u}}$
			\State Compute $w(\bm g)=\frac{\rho_B(\bm g)}{1+\kappa d(\bm g_{\text{UAV}},\bm g)}$ for all $\bm g\in\mathcal V_{\text{u}}$
			\State Sample one target $\bm g^{\star}$ with probability $p(\bm g)=w(\bm g)/\sum_{\tilde{\bm g}\in\mathcal V_{\text{u}}}w(\tilde{\bm g})$

			\vspace{0.05in}
			\State \textbf{Stage 2: Trajectory Generation via \ac{uaps}}
			\State Normalize $\bm U$ to $\tilde{\xi}(\cdot)\in[0,1]$ and set $\alpha=1/(1+\lambda)$
			\State Initialize $r(\bm g_{\text{UAV}})=0$ and $r(\bm g)=\infty$ for all other grid cells
			\State Initialize open set $\mathcal O=\{\bm g_{\text{UAV}}\}$, closed set $\mathcal C=\emptyset$, and parent map $\pi(\cdot)$
			\While{$\mathcal O\neq\emptyset$}
			\State Extract $\bm g=\arg\min_{\bm q\in\mathcal O}\{r(\bm q)+(1-\alpha)d(\bm q,\bm g^{\star})\}$
			\If{$\bm g=\bm g^{\star}$}
			\State \textbf{break}
			\EndIf
			\State Move $\bm g$ from $\mathcal O$ to $\mathcal C$
			\ForAll{four-neighbor cells $\bm g'\notin\mathcal C$ of $\bm g$}
			\State Compute $c_{\text{step}}(\bm g,\bm g')=1-\alpha\,\tilde{\xi}(\bm g')$
			\If{$r(\bm g)+c_{\text{step}}(\bm g,\bm g')<r(\bm g')$}
			\State Set $r(\bm g')=r(\bm g)+c_{\text{step}}(\bm g,\bm g')$ and $\pi(\bm g')=\bm g$
			\State Insert or update $\bm g'$ in $\mathcal O$
			\EndIf
			\EndFor
			\EndWhile
			\State Reconstruct $\mathcal Q=(\bm g_{\text{UAV}},\ldots,\bm g^{\star})$ by backtracking the parent map $\pi(\cdot)$
			\State \Return \(\bm{g}^{\star}, \mathcal{Q}\)
		\end{algorithmic}
	\end{algorithm}

	\section{Simulation Results\label{sec:Simulation-Results}}

	In this section, we present simulation results to evaluate the
	proposed flow-based active learning
	framework.

	The ground truth radio maps were generated using the Sionna ray-tracing
	simulator \cite{FayJakMer:J25}, which provides physically accurate
	radio channel predictions. We adopt the open-source ``Etoile'' 3D
	urban scenario, and randomly deploy $7$ transmitters within this
	scenario. The ground truth radio map is then computed over a planar
	grid of $128\times128$ at a fixed altitude of $60$ m, yielding an
	RSS matrix $\bm{H}$ with dimensions $I=J=128$. This process is repeated
	to generate a dataset of $2000$ radio maps for training and $200$
	radio maps for testing. In the training stage of flow matching, the
	source distribution $\mathcal{P}_{0}$ is chosen as $\mathcal{N}(0,1)$.
	Unless otherwise specified, the UAV measurements are corrupted by
	additive Gaussian noise with variance $\sigma_y^2=1$.

	The performance of the reconstruction is quantified using the \ac{nmse}.
	To ensure physical consistency, the NMSE is calculated in the linear
	power domain rather than the dB scale. The NMSE is then computed as
	\[
	\text{NMSE}=\frac{\|10^{\bm{H}/10}-10^{\bm{Z}/10}\|_{\text{F}}^{2}}{\|10^{\bm{H}/10}\|_{\text{F}}^{2}},
	\]
	where $\bm{H}$ and $\bm{Z}$ are the ground truth and reconstructed
	radio map, $10^{\bm{H}/10}$ and $10^{\bm{Z}/10}$ denote element-wise
	conversions from the dB scale to the linear power scale, not matrix
	exponentiation.

	\subsection{Radio Map Construction Performance}

	We compare the proposed TFM-PnP approach with several representative
	baselines as follows:
	\begin{itemize}
		\item PnP-Flow \cite{MarGagHag:J25}: A flow-based prior is directly
		incorporated into a \ac{pnp} framework, with the inference time starting from $t=0$.
		\item Diffusion posterior sampling (DPS) \cite{chung2024diffusionposteriorsamplinggeneral}:
		A diffusion-based inverse solver that reconstructs the unknown field
		by iteratively sampling from the posterior distribution conditioned
		on the observed measurements. DPS uses a generative prior for stable recovery but usually requires more sampling steps.
		\item RePaint \cite{lugmayr2022repaintinpaintingusingdenoising}: A \ac{ddpm}
		based inpainting approach. It alternates between denoising and conditional
		resampling to gradually fill in the missing regions of the map, serving
		as a representative diffusion-based inpainting baseline.
		\item Deep radio map and uncertainty estimator (DRUE) \cite{ShrRomChe:J23}: An autoencoder-based architecture is employed to learn the underlying radio-map structure; for the reconstruction-only comparison, we use its radio-map prediction component.
		\item Gaussian process regression (GPR): A classical probabilistic interpolation
		baseline that reconstructs the radio map by learning spatial correlations
		from the observed measurements. In our implementation, GPR is trained
		on the normalized spatial coordinates and the corresponding normalized
		signal values, using a composite kernel consisting of a constant kernel,
		a radial basis function (RBF) kernel, and a white-noise kernel. The kernel hyperparameters
		are optimized via marginal likelihood with multiple restarts, enabling
		adaptive modeling of the map smoothness and noise level.
		\item \Ac{knn}, with $k=5$. A distance-weighted interpolation method is
		used as a lightweight local baseline. It predicts each missing location
		based on the observed values of its five nearest neighbors in the
		normalized spatial domain.
	\end{itemize}

	\subsubsection{Sensitivity with Respect to $\eta$}

	In this subsection, we investigate the sensitivity of the proposed
	TFM-PnP framework to the gradient step size $\eta$. We adopt a constant
	step size, i.e., $\eta_k= \eta$ for all iterations $k$. We
	use a finite-step implementation in which the constant step
	size is tuned empirically. This differs from the scaled step-size regime
	used in the idealized asymptotic analysis, but follows the common practice
	of finite-step PnP-type inference algorithms.
	To evaluate this choice, we
	consider a highly sparse observation scenario where grid locations
	are uniformly sampled at random, with sampling rates ranging from
	$2\%$ to $6\%$. The reconstruction performance of the proposed method
	is evaluated across different gradient step sizes,
	specifically $\eta\in\{0,0.5,1,1.5,2,2.2\}$.

	The corresponding NMSE curves in Fig.~\ref{fig:versus_eta}
	show that increasing the gradient step size improves reconstruction accuracy, with the best performance obtained around $\eta=2$.
	When $\eta=0$ (i.e., no data-fidelity gradient update is applied),
	the NMSE remains at a high level. However, as $\eta$ increases from
	$0.5$ to $2$, we observe a consistent and marked reduction in NMSE
	across all sampling rates. Notably, the performance gain begins to
	saturate as the step size becomes larger, with the curves for $\eta=2$
	and $\eta=2.2$ showing only marginal differences. These results indicate that a sufficiently large gradient step size is important for enforcing data consistency, and $\eta=2$ provides near-best performance in our experiments.
	\begin{figure}
		\includegraphics{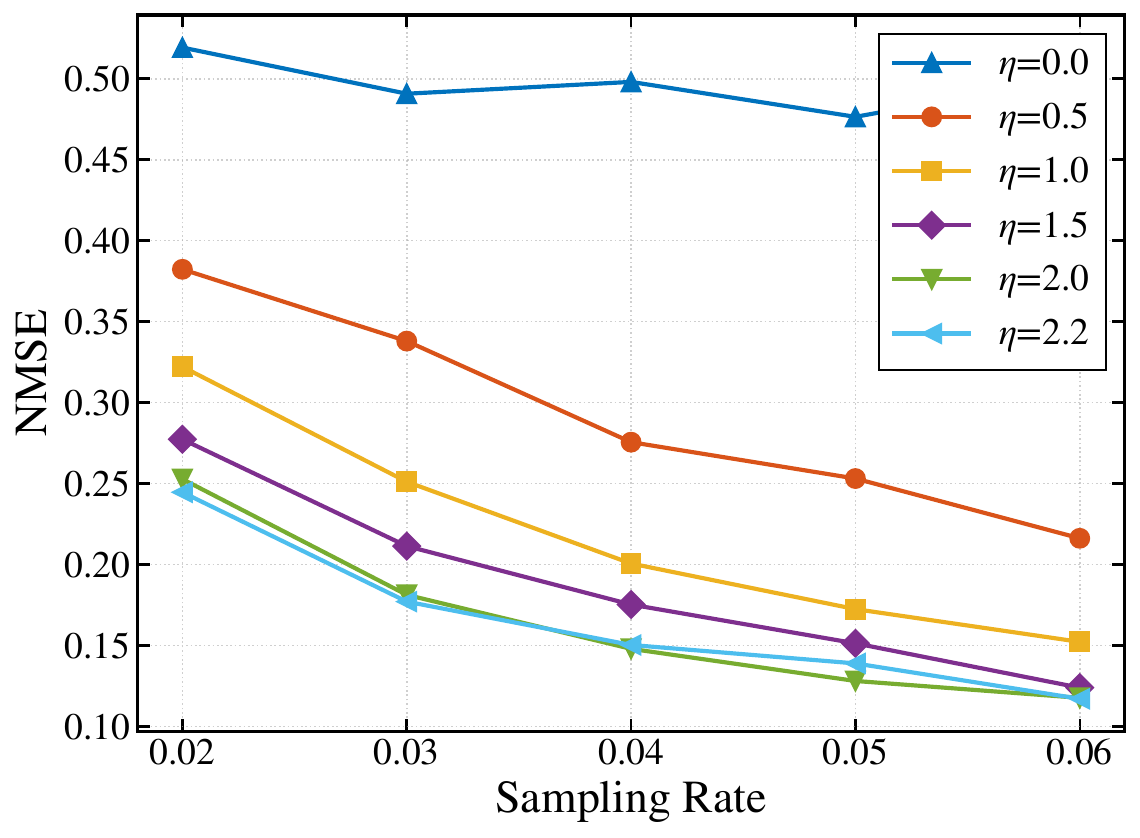}\caption{\label{fig:versus_eta}Performance comparison (NMSE) across varying
			sampling rates under different gradient step sizes $\eta$. A larger gradient step size
			lowers the NMSE, though the marginal performance gain diminishes as
			$\eta$ exceeds 2.0.}
	\end{figure}

	\subsubsection{Effectiveness of the Time-Truncation Mechanism}

	\begin{figure}
		\subfigure{\includegraphics[width=1\columnwidth]{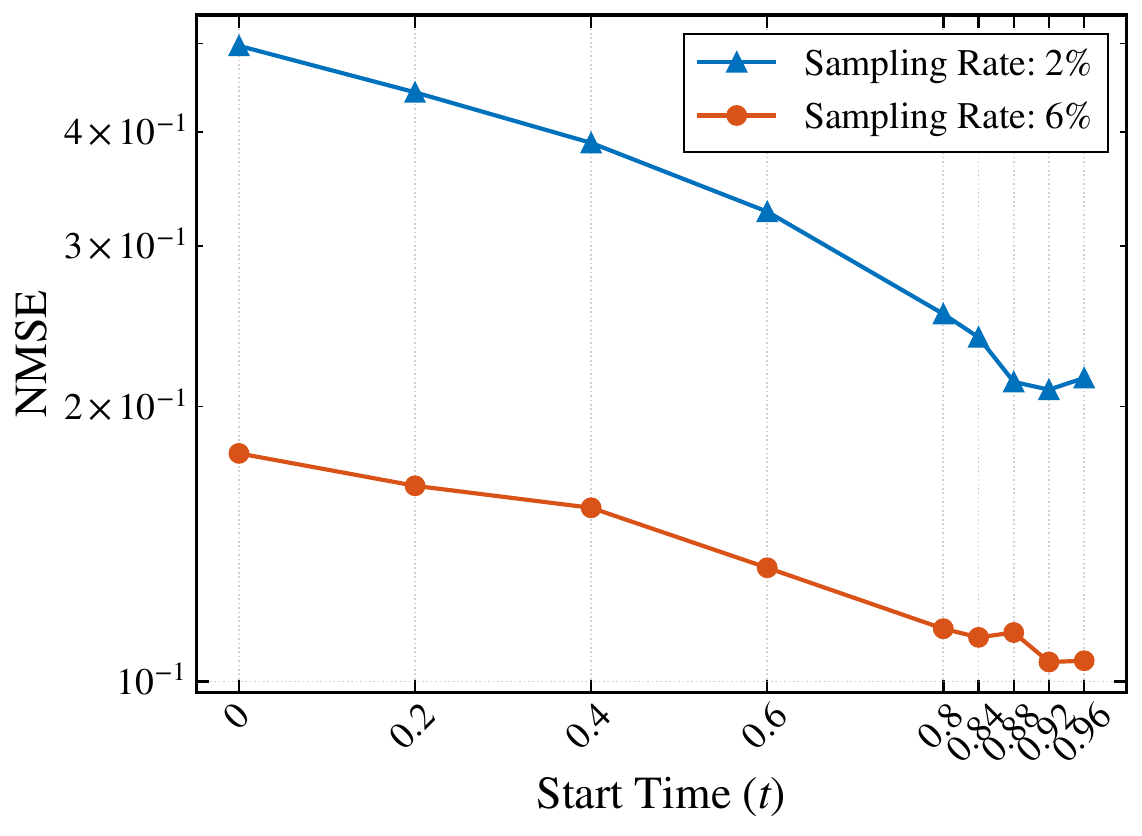}}\caption{\label{fig:fix_timestep}Performance comparison (NMSE) under various
			start times. Increasing the start time lowers the NMSE
			initially, while the performance improvement becomes marginal and
			saturates around $t=0.9$.}
	\end{figure}

	To investigate the influence of the start time $a$ on the performance
	of the proposed TFM-PnP framework, we vary $a$ from $0$ to
	$0.96$. As illustrated in Fig.~\ref{fig:fix_timestep}, increasing
	the start time $a$ initially leads to a steady
	reduction in NMSE across both $2\%$ and $6\%$ sampling rates. However,
	as $a$ approaches $0.9$, the performance gain becomes marginal,
	indicating that the reconstruction quality begins to saturate. The
	framework achieves stable performance around $a=0.9$. Increasing
	the sampling rate also improves reconstruction quality, since denser
	measurements provide stronger constraints and reduce uncertainty.

	These findings suggest that skipping the early flow stages, which are dominated by substantial noise and limited structural information, allows the PnP process to focus on refining fine spatial structures rather than performing coarse denoising.

	\subsubsection{Reconstruction Performance Comparison}

	We evaluate the reconstruction performance by varying the sampling
	ratio from $2$\% to $6$\%. As shown in Fig.~\ref{fig:nmse_sr_cmp_baselines},
	the proposed TFM-PnP approach outperforms all baseline
	methods, achieving more than $10$\% improvement in NMSE across all
	sampling ratios.

	Although all baseline methods are carefully
	tuned, they perform worse than the proposed
	approach. Generative baselines such as RePaint, DPS, and PnP-Flow typically begin reconstruction from pure noise. Under highly sparse sampling conditions (e.g., $2\%$ to $6\%$),
	the available measurements are often insufficient to reliably guide
	the stochastic and noise-dominated early stages of generation,
	leading to structural hallucinations or deviations from the
	true spatial layout. In contrast, the proposed TFM-PnP framework mitigates these limitations
	through an early-time truncation.

	To provide a visual comparison and assess the generalization
	capability of our method, we fix the sampling ratio at $6\%$ and
	present representative reconstruction results across two distinct
	scenarios in Fig.~\ref{fig:nmse_sr_cmp_baselines_visible} and Fig.
	\ref{fig:nmse_sr_cmp_baselines_visible-1}. While traditional baselines
	(e.g., GPR, DRUE) suffer from severe over-smoothing, generative
	baselines like DPS and PnP-Flow show competitive performance,
	recovering the overall power distribution. However, upon
	closer inspection of complex structural regions, the proposed approach
	better preserves fine-grained details and sharp boundaries, yielding
	more accurate and spatially consistent reconstructions. 

	\begin{figure}
		\includegraphics[width=1\columnwidth]{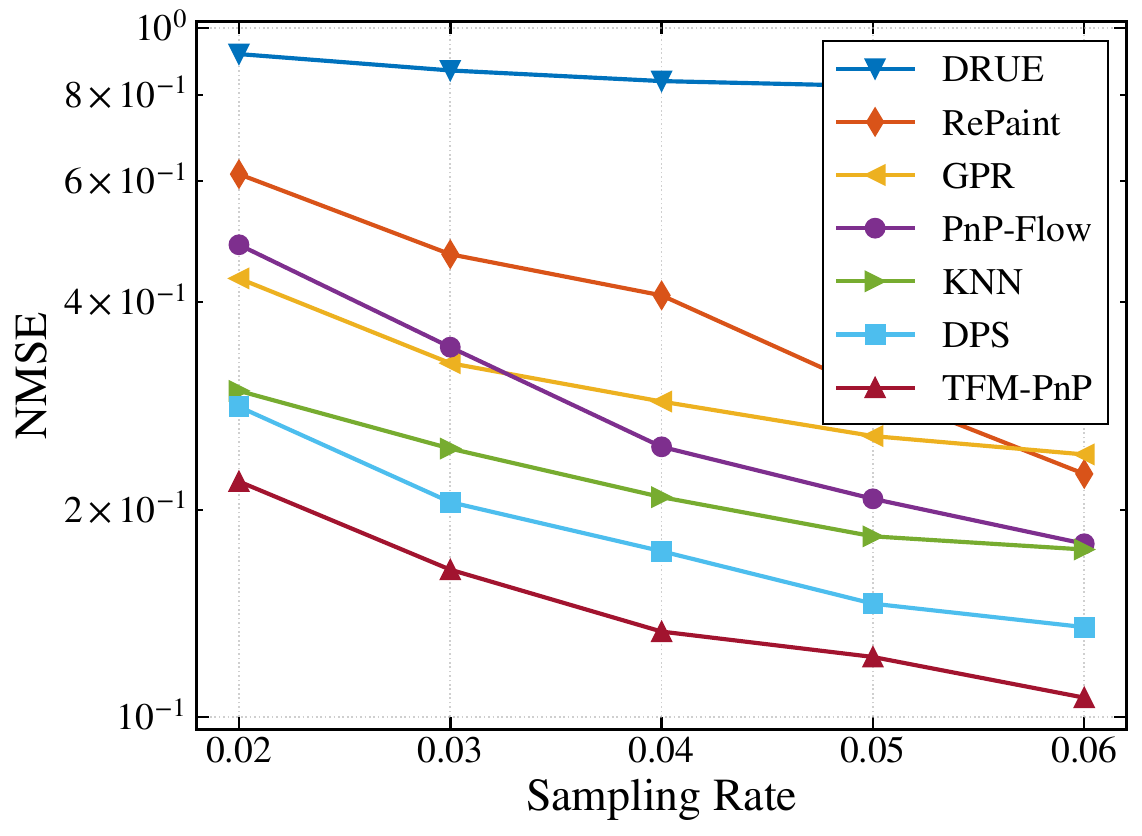}\caption{\label{fig:nmse_sr_cmp_baselines}Reconstruction NMSE versus different
			sampling ratios. The proposed TFM-PnP achieves the lowest NMSE across
			all sampling ratios.}
	\end{figure}
	\begin{figure}
		\includegraphics[width=1\columnwidth]{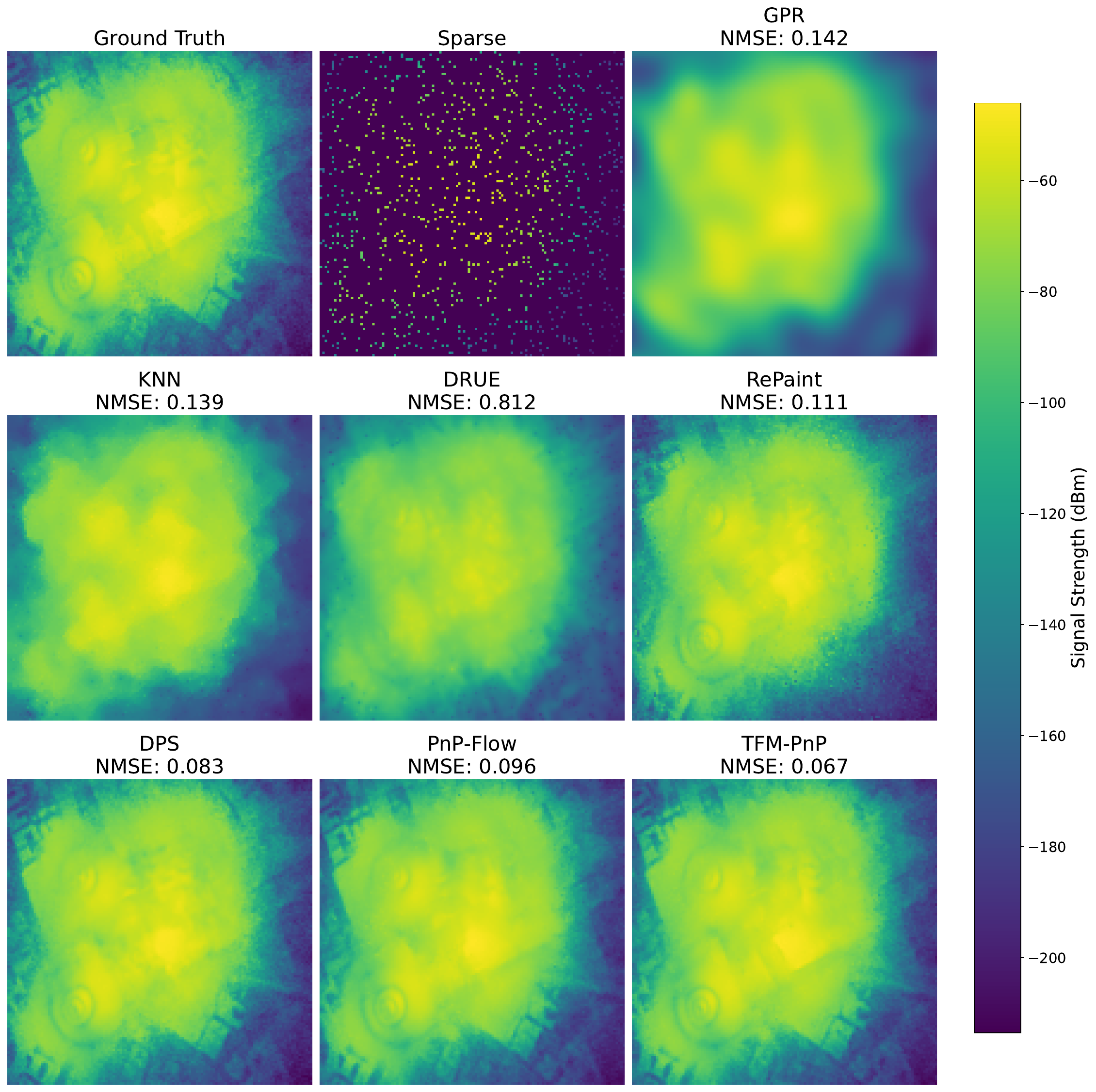}\caption{\label{fig:nmse_sr_cmp_baselines_visible}Visual comparison under
			$6\%$ sampling ratio. Compared with the baselines, the proposed TFM-PnP
			yields more accurate and spatially consistent reconstructions.}
	\end{figure}

	\subsubsection{Running Time and Inference Steps}

To comprehensively evaluate the computational efficiency of the proposed
	TFM-PnP framework, this subsection investigates the impact of the
	number of inference steps on reconstruction accuracy and quantitatively
	compares the execution time against existing baselines. Fig.~\ref{fig:mse_steps}
	illustrates the NMSE of TFM-PnP with respect to the number of inference
	steps, evaluated across various sparse sampling rates. A consistent
	trend is observed: the NMSE decreases rapidly as the number of inference
	steps increases to approximately $50$. Subsequently, as the integration
	steps exceed $60$, the marginal performance gain becomes small,
	and the NMSE curves saturate. 

	Building upon these step-wise observations, we further quantify the
	corresponding wall-clock inference time. Based on the saturation point
	identified in Fig.~\ref{fig:mse_steps}, the number of inference
	steps for the proposed TFM-PnP is empirically set to $60$. As detailed
	in Table~\ref{tab:running_time}, the proposed framework requires less
	inference time than the competing generative solvers. Note that the results are reported as the mean
	$\pm$ standard deviation across multiple independent trials to ensure
	statistical reliability. Diffusion models, such as RePaint and DPS, require
	hundreds of inference steps, which leads to substantial
	inference latency and makes them less suitable for power-constrained
	UAV applications. In contrast, TFM-PnP reduces the
	required forward passes, requiring only 0.74 seconds per radio map.
	This lightweight inference process makes the proposed
	framework suitable for real-time or near-real-time UAV tasks that require
	fast environmental sensing and trajectory adaptation.
		\begin{figure}
		\includegraphics[width=1\columnwidth]{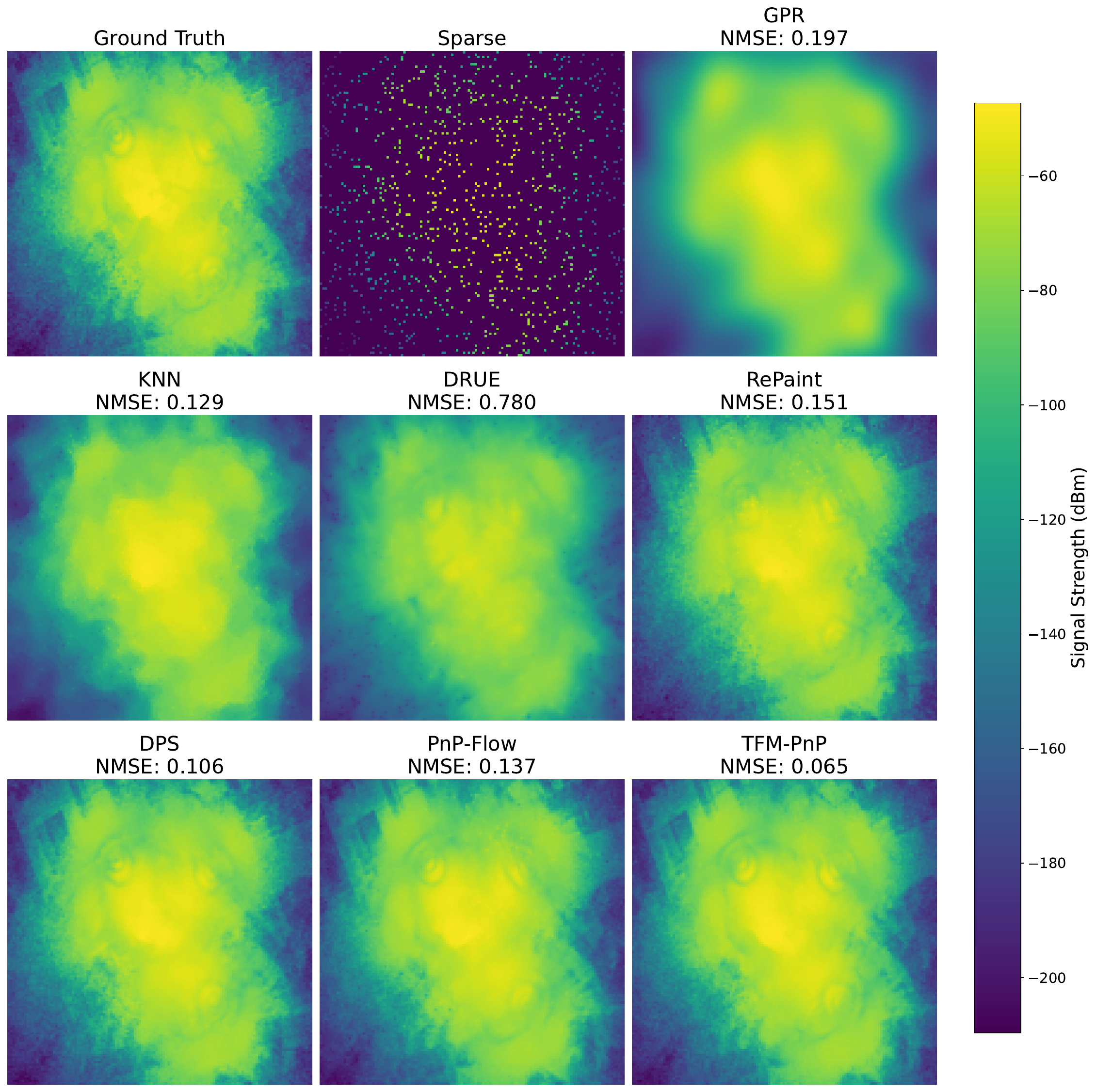}\caption{\label{fig:nmse_sr_cmp_baselines_visible-1}Cross-scenario validation
			of radio map reconstruction under a $6\%$ sampling ratio in a topologically
			distinct environment. The low NMSE and spatial
			consistency indicate the generalization capability of the
			proposed TFM-PnP framework.}
	\end{figure}
		\begin{table}[t]
		\caption{Average Inference Time Per Radio Map}
		\label{tab:running_time}
		\centering
		\begin{tabular}{l c}
			\hline\hline
			\textbf{Method} & \textbf{Average Inference Time (s)} \\
			\hline
			GPR & $8.1234 \pm 0.0958$ \\
			RePaint & $6.9604 \pm 0.0118$ \\
			DPS & $5.1105 \pm 0.0382$ \\
			PnP-Flow & $1.0805 \pm 0.0002$ \\
			\textbf{TFM-PnP (Proposed)} & $\bm{0.7359 \pm 0.0412}$ \\
			KNN & $0.0715 \pm 0.0394$ \\
			DRUE & $0.0016 \pm 0.0001$ \\
			\hline\hline
		\end{tabular}
	\end{table}

	\begin{figure}
		\includegraphics[width=1\columnwidth]{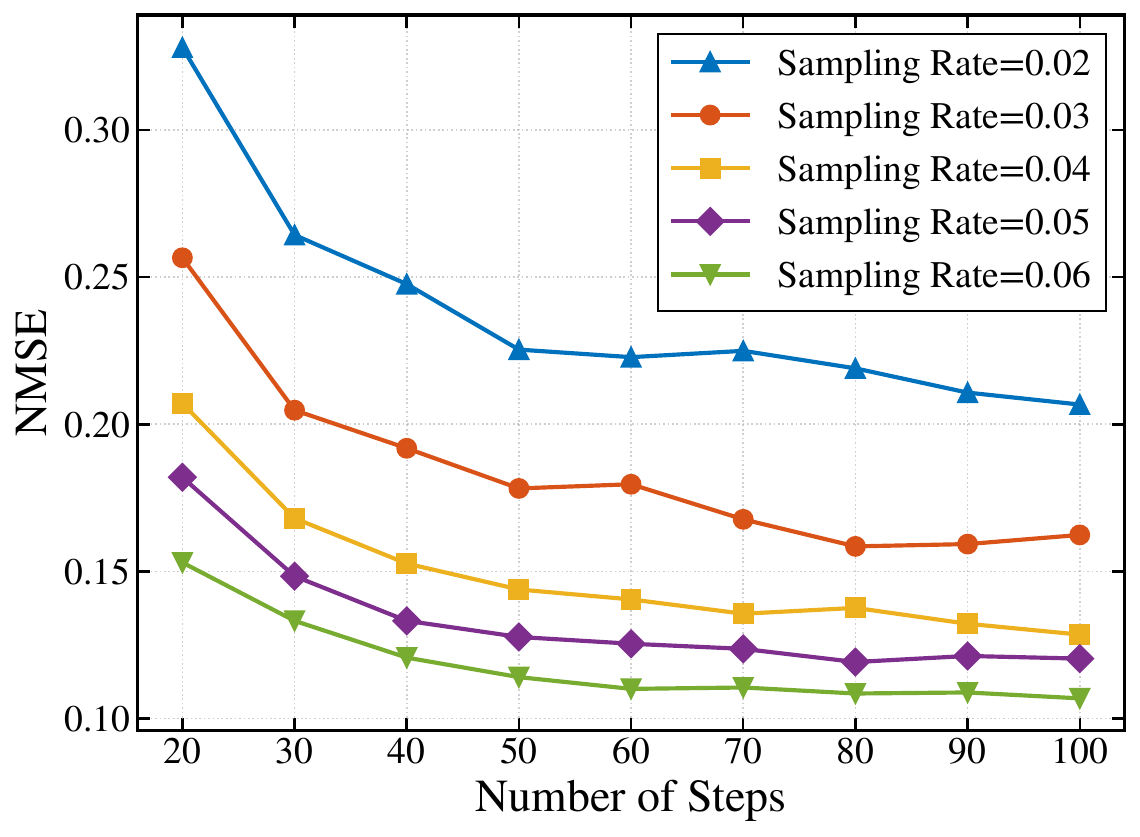}\caption{\label{fig:mse_steps}Impact of the number of inference steps on the
			reconstruction NMSE of the proposed TFM-PnP framework under various
			sampling rates. Setting the number of inference steps to 60 is
			sufficient to achieve near-saturated performance.}
	\end{figure}
	\begin{figure}
	\includegraphics[width=1\columnwidth]{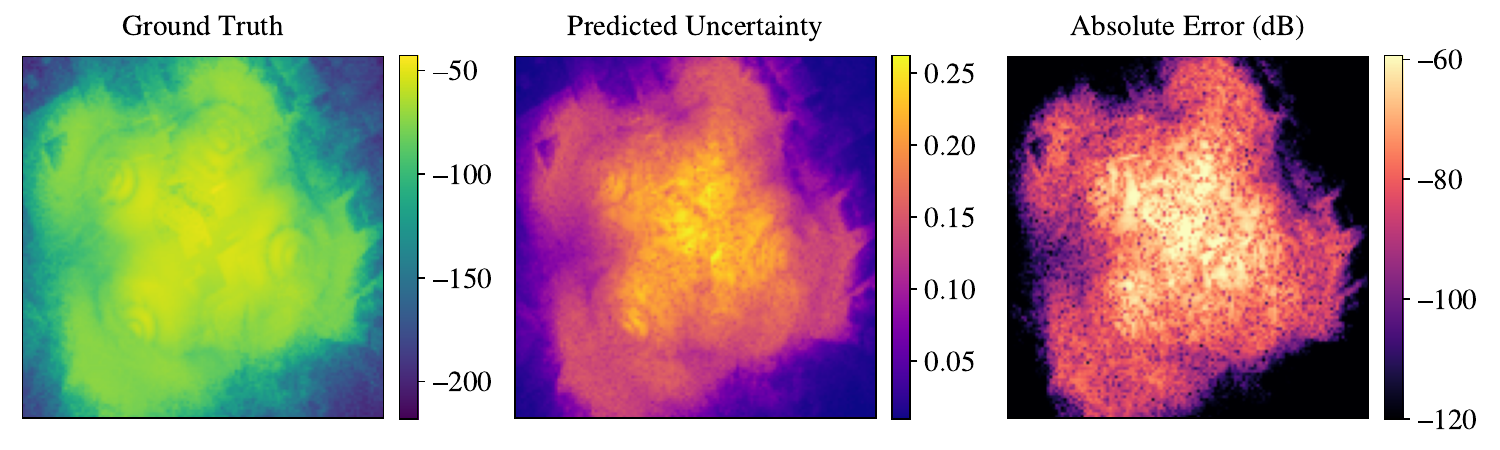}\caption{\label{fig:uncertainty_map}Comparison of uncertainty maps and reconstruction
		errors. From left to right: Ground truth radio map, predicted uncertainty
		map of the proposed method, and the absolute reconstruction error
		(dB). The proposed method produces an uncertainty map that is
		qualitatively consistent with the spatial distribution of the true reconstruction
		error and the underlying radio map structure. This suggests the
		model's ability to indicate regions of high predictive
		uncertainty.}
\end{figure}
	\subsection{Active Learning and Radio Map Construction}

	In this experiment, we assume that $2\%$ of the grid locations are
	initially observed. During the active learning phase, the UAV is assigned
	an additional flight budget of $2,000$ movement steps, guided by the
	proposed active learning strategy summarized in Algorithm~\ref{alg:active_uaps}.
	At each visited grid cell, the UAV acquires one RSS measurement, and
	$\mathcal{H}_{\Omega}(\bm{H})$ extracts the entries of the radio map
	corresponding to the set of visited locations. In the active learning
	phase, we set $\kappa=1$, $\lambda=1$, and $B=5$.

	We benchmark the proposed method against two baselines under the same overall UAV flight budget.

	\begin{itemize}
		\item Baseline 1: Random Sampling. In this baseline, $2,000$ additional
		sampling locations are selected uniformly at random from the entire
		grid. The reconstruction is then performed through the proposed TFM-PnP
		based on all collected samples. This strategy does not involve an
		iterative, uncertainty-guided trajectory planning process.
		\item Baseline 2: Deep Radio Map and Uncertainty Estimator (DRUE) \cite{ShrRomChe:J23}.
		This method employs two autoencoders to estimate the radio map and
		its uncertainty separately, followed by a UAV trajectory designed
		in an uncertainty-aware manner.
	\end{itemize}

	\subsubsection{Uncertainty Map}

	The uncertainty map $\bm{U}^{(\tau)}$ is obtained by generating $M$ reconstructed radio maps and computing the pixel-wise variance
	across the samples. This approach provides a practical proxy for reconstruction uncertainty induced by sparse measurements. 

	Fig.~\ref{fig:compare_M} presents a sensitivity analysis regarding
	the number of reconstruction samples $M$, evaluated from 3 to 10. While all configurations
	demonstrate a consistent decrease in NMSE over the flight steps, the
	curves remain closely aligned throughout the process. The small performance variation across different values of $M$ suggests that the method is not sensitive to this parameter over the tested range. To maintain sufficient statistical
	confidence without introducing unnecessary computational overhead,
	we adopted $M=5$ as the default setting for the remainder of this
	study.

	As shown in Fig.~\ref{fig:uncertainty_map}, the proposed method produces
	uncertainty patterns that are well aligned with the physical characteristics
	of the radio map. The uncertainty map exhibits a strong qualitative correspondence with the actual reconstruction error, suggesting that the sample variance provides a useful proxy for reconstruction uncertainty. Areas of high uncertainty
	coincide with higher reconstruction errors, whereas regions of low
	uncertainty correspond to accurate predictions. Such alignment is
	crucial for downstream tasks such as active sampling and adaptive
	flight planning, where uncertainty serves as a guide for selecting
	informative measurements.

	\subsubsection{Active Learning}
			\begin{figure}
	\centering\includegraphics[width=1\columnwidth]{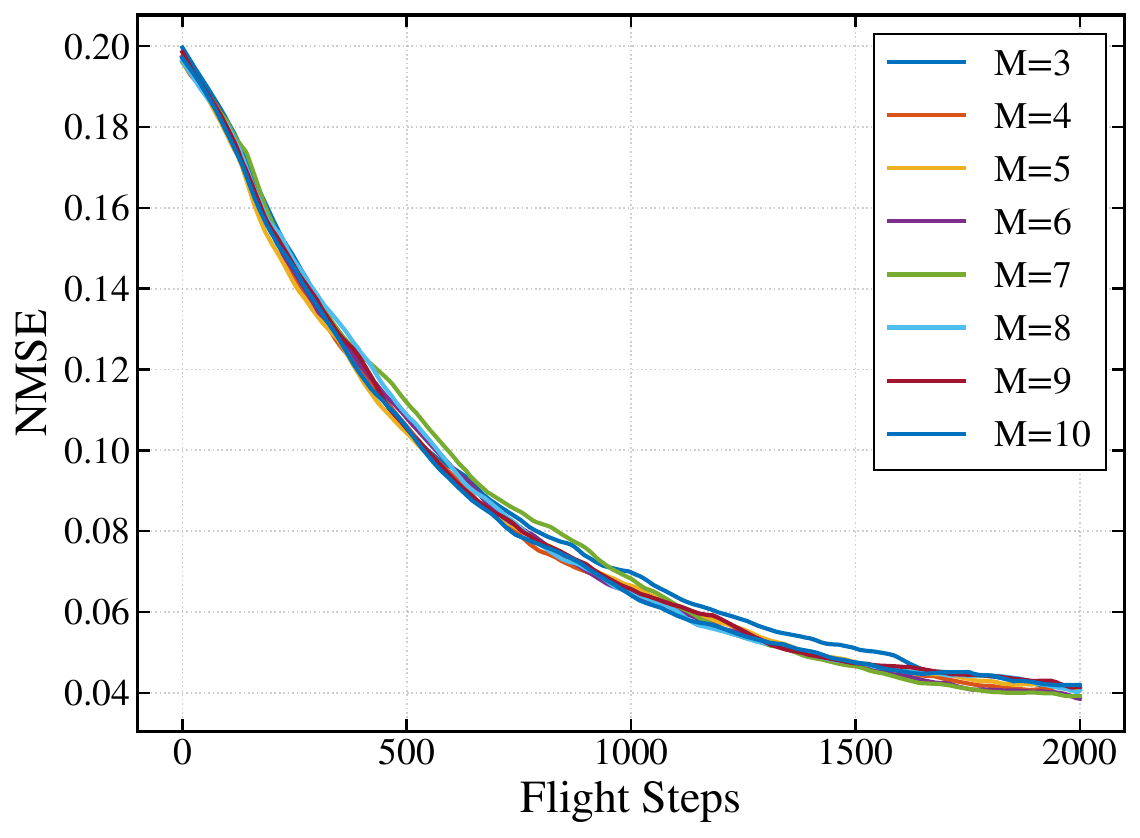}
	
	\caption{\label{fig:compare_M}Sensitivity analysis of reconstruction performance (NMSE) with respect to the number of reconstruction samples $M$. The overlapping curves indicate limited sensitivity to variations in $M$.}
\end{figure}
\begin{figure}
	\centering\includegraphics[width=1\columnwidth]{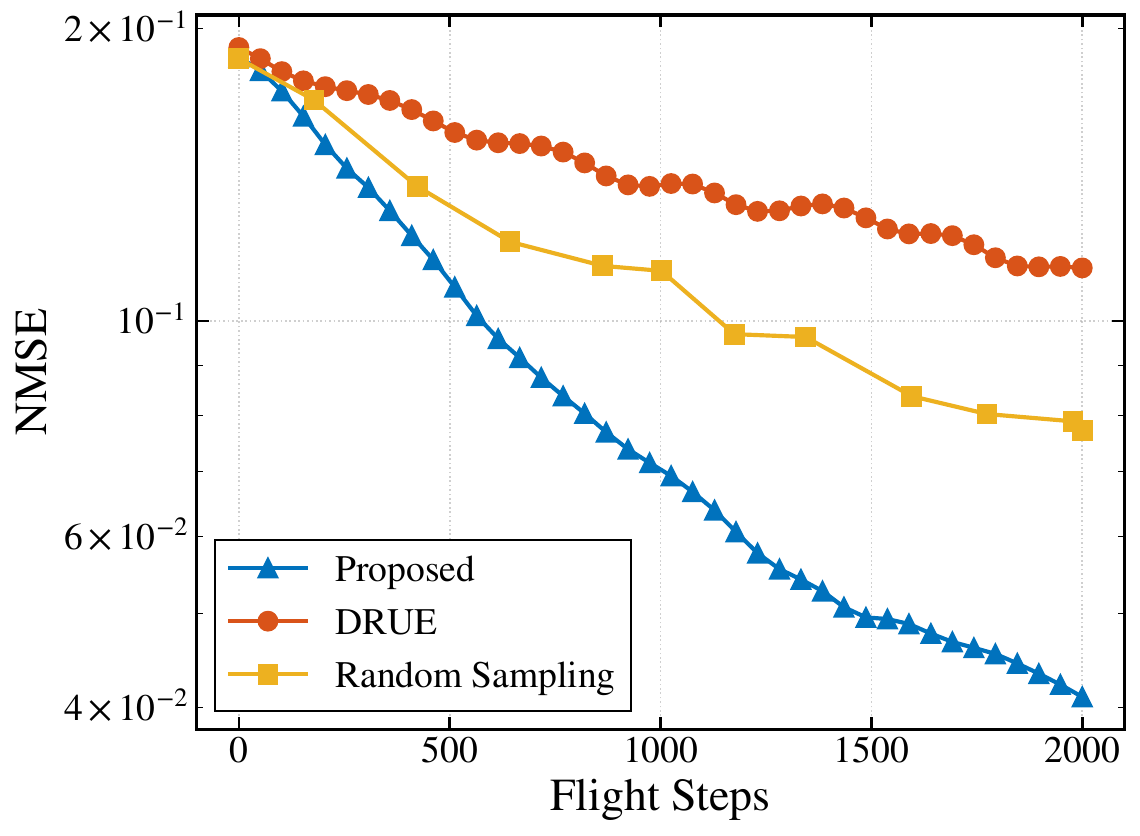}
	
	\caption{\label{fig:Reconstruction-NMSE-versus sampling process}NMSE with
		respect to UAV flight steps, showing that the proposed uncertainty-guided
		sampling strategy achieves faster error reduction.}
\end{figure}
	To assess the effectiveness of the proposed active learning strategy
	in UAV-assisted low-altitude radio map reconstruction, we compare
	the proposed method with two baselines, namely random sampling and DRUE.

	The quantitative performance is presented in Fig.~\ref{fig:Reconstruction-NMSE-versus sampling process},
	which plots the NMSE against the total UAV flight steps. Throughout
	the active learning process, the proposed method achieves
	a faster NMSE reduction than DRUE and random sampling.
	This improvement arises from the uncertainty-driven sequential target-selection
	mechanism, which identifies informative and reachable sensing locations,
	and from the utility-aware path planner, which favors uncertain regions
	along the transit path while maintaining feasible flight trajectories.

	After $2,000$ flight steps, the proposed approach achieves an NMSE
	reduction exceeding $50\%$, highlighting
	a performance gain in both reconstruction accuracy and
	data efficiency. These results suggest that the proposed
	uncertainty-driven trajectory planning framework enables the UAV to
	acquire informative samples within a limited flight budget,
	yielding accurate and spatially consistent radio maps.

	\section{Conclusion\label{sec:Conclusion}}

	In this paper, we proposed an active learning framework based on flow
	matching for efficient and accurate low-altitude radio map construction from
	sparse UAV measurements. To address the computational bottlenecks and instability
	of traditional generative solvers, we developed the TFM-PnP approach, which
	integrates a continuous-time flow matching prior into a \ac{pnp}
	architecture. We first established an idealized ODE-method characterization of the late-stage asymptotic behavior of the generative \ac{pnp} dynamics. This result provides a dynamical interpretation of the late-stage refinement process after truncation. 
	Combined with an explicit data-fidelity gradient update and flow-path
	reparameterization, TFM-PnP achieves high-fidelity and physically consistent
	reconstruction without the need for model retraining. Furthermore, we used the inherent generative diversity of the
	flow model to derive a spatial uncertainty map. This pixel-wise
	uncertainty systematically guides a sequential trajectory-planning strategy,
	where an informative and reachable target location is selected at each
	replanning round and connected using the proposed \ac{uaps} algorithm. Extensive evaluations on Sionna ray-tracing datasets show that the proposed framework achieves NMSE reductions relative to the compared baselines while maintaining strong spatial fidelity.


	\appendices{}


	\section{Proof of Proposition \ref{prop:ode_limit}\label{sec:Proof-of-Proposition_ode}}

Fix a finite index \(k_0\) and consider the sequence for \(k\ge k_0\).
	Any finite number of early updates, including the possible initial update
	from \(t_0=0\), is treated as an initialization transient and does not
	affect the asymptotic ODE limit. For notational simplicity, we keep the
	index \(k\) for this late-stage sequence and define \(h_k\triangleq1-t_k\).
	From \(\sum_{k=k_0}^{\infty}h_k=\infty\) and
	\(\sum_{k=k_0}^{\infty}h_k^2<\infty\), we have \(h_k\to0\) and hence
	\(t_k\to1\).

	Let \(\{\mathcal F_k\}_{k\ge k_0}\) be the filtration generated by the
	past iterates and the reference samples drawn before step \(k\). By assumption, the fresh sample
	\(\bm Z_0^{(k)}\) is independent of \(\mathcal F_k\),
	\[
	\mathbb E[\bm Z_0^{(k)}\mid\mathcal F_k]=\bm 0,\qquad
	\mathbb E[\|\bm Z_0^{(k)}\|^2\mid\mathcal F_k]\le C_0
	\]
	for some constant \(C_0>0\). Using the notation in
	\eqref{eq:pnp_flow_iteration}, the measurement-consistency update is
	\begin{equation}
		\bar{\bm Z}^{(k+1)}
		=
		\bm Z^{(k)}-\eta_k\nabla F(\bm Z^{(k)}).
		\label{eq:proof_bar_update}
	\end{equation}
	The reparameterized state is
	\begin{equation}
		\bm W^{(k+1)}
		=
		t_k\bar{\bm Z}^{(k+1)}
		+
		h_k\bm Z_0^{(k)} .
		\label{eq:proof_w_update}
	\end{equation}
	Finally, applying the flow-induced prior operator gives
	\begin{equation}
		\bm Z^{(k+1)}
		=
		\bm W^{(k+1)}
		+
		h_k\bm v_{\theta}(t_k,\bm W^{(k+1)}).
		\label{eq:proof_z_update}
	\end{equation}

	We first derive the one-step increment. Substituting
	\eqref{eq:proof_bar_update} into \eqref{eq:proof_w_update} and using
	\(\beta_k=t_k\eta_k/h_k\), we obtain
	\begin{equation}
		\bm W^{(k+1)}
		=
		\bm Z^{(k)}
		+
		h_k\Big(
		\bm Z_0^{(k)}
		-\bm Z^{(k)}
		-\beta_k\nabla F(\bm Z^{(k)})
		\Big).
		\label{eq:proof_w_increment}
	\end{equation}
	Therefore, subtracting \(\bm Z^{(k)}\) from \eqref{eq:proof_z_update}
	yields

\begin{align}
	&\bm Z^{(k+1)}-\bm Z^{(k)} \nonumber \\
	=\  &	h_k\Big(
	\bm Z_0^{(k)}
	-\bm Z^{(k)}
	-\beta_k\nabla F(\bm Z^{(k)})
	+
	\bm v_{\theta}(t_k,\bm W^{(k+1)})
	\Big).
	\label{eq:proof_final_increment}
\end{align}
	Next, we replace the velocity evaluated at \(\bm W^{(k+1)}\) by the
	velocity evaluated at the current iterate. Define
	\[
	\bm R^{(k)}
	\triangleq
	\bm v_{\theta}(t_k,\bm W^{(k+1)})
	-
	\bm v_{\theta}(t_k,\bm Z^{(k)}).
	\]
	By the boundedness assumption, the iterates for \(k\ge k_0\) lie in a compact set
	\(\mathcal K\subset\mathbb R^{I\times J}\) almost surely. Since
	\(\nabla F\) is locally Lipschitz, it is bounded on \(\mathcal K\).
	Moreover, \(\{\beta_k\}\) is bounded because \(\beta_k\to\beta\).
	From \eqref{eq:proof_w_increment}, for some finite constant
	\(C_1>0\),
	\begin{equation}
		\mathbb E\!\left[
		\|\bm W^{(k+1)}-\bm Z^{(k)}\|^2
		\mid \mathcal F_k
		\right]
		\le C_1h_k^2 .
		\label{eq:proof_state_difference_bound}
	\end{equation}
	By the uniform Lipschitz property of \(\bm v_{\theta}\) in \(\bm Z\),
	there exists \(L>0\) such that
	\[
	\|\bm R^{(k)}\|
	\le
	L\|\bm W^{(k+1)}-\bm Z^{(k)}\|.
	\]
	Consequently,
	\begin{equation}
		\mathbb E\!\left[\|\bm R^{(k)}\|^2\mid\mathcal F_k\right]
		\le
		L^2C_1h_k^2 .
		\label{eq:proof_residual_bound}
	\end{equation}

	Now define the limiting drift
	\[
	\mathcal G(\bm Z)
	\triangleq
	-\bm Z-\beta\nabla F(\bm Z)+\bm v_{\theta}(1,\bm Z).
	\]
	Substituting
	\(\bm v_{\theta}(t_k,\bm W^{(k+1)})
	=\bm v_{\theta}(t_k,\bm Z^{(k)})+\bm R^{(k)}\)
	into \eqref{eq:proof_final_increment}, the recursion becomes
	\begin{equation}
		\bm Z^{(k+1)}
		=
		\bm Z^{(k)}
		+
		h_k\big(\mathcal G(\bm Z^{(k)})+\bm Z_0^{(k)}\big)
		+
		\bm\varepsilon^{(k)},
		\label{eq:proof_sa_form}
	\end{equation}
	where
\begin{align}
	\bm\varepsilon^{(k)}
	&=
	h_k\Big(
	-(\beta_k-\beta)\nabla F(\bm Z^{(k)}) \nonumber \\
	&\quad
	+ \bm v_{\theta}(t_k,\bm Z^{(k)})
	- \bm v_{\theta}(1,\bm Z^{(k)})
	\Big)
	+ h_k\bm R^{(k)} .
	\label{eq:proof_error_term}
\end{align}

	We next show that the error term is negligible on the ODE time scale.
	Since \(\beta_k\to\beta\), \(t_k\to1\), and
	\(\bm v_{\theta}(t,\bm Z)\) is continuous in \(t\) and uniformly
	Lipschitz in \(\bm Z\), the convergence
	\(\bm v_{\theta}(t_k,\cdot)\to\bm v_{\theta}(1,\cdot)\) is uniform
	on the compact set containing the iterates. Hence, the bracketed term in
	\eqref{eq:proof_error_term} converges to zero almost surely. In addition,
	\eqref{eq:proof_residual_bound} implies \(\bm R^{(k)}=\mathcal O(h_k)\)
	in \(L^2\). Hence,
	\begin{equation}
		\bm\varepsilon^{(k)}=h_k\bm\delta_k+\mathcal O_{L^2}(h_k^2),
		\qquad
		\bm\delta_k\to\bm 0 \quad \text{a.s.}
		\label{eq:proof_error_small}
	\end{equation}
	Thus, over any bounded interval of the algorithmic time \(s\), the
	cumulative contribution of \(\bm\varepsilon^{(k)}\) vanishes as the
	starting index tends to infinity. Indeed, over any interval with
	\(\sum h_k\le T<\infty\), the first term is bounded by
	\(T\sup_{k\ge n}\|\bm\delta_k\|\), while the second term is controlled
	by the remainder of \(\sum h_k^2\) in \(L^2\).

	It remains to control the cumulative stochastic perturbation. Define
	the martingale
	\[
	\bm M_n
	\triangleq
	\sum_{k=k_0}^{n-1}h_k\bm Z_0^{(k)} .
	\]
	Because \(\bm Z_0^{(k)}\) is a martingale difference with uniformly
	bounded second moment,
	\[
	\mathbb E\|\bm M_n-\bm M_m\|^2
	\le
	C_0\sum_{k=m}^{n-1}h_k^2,\qquad n>m\ge k_0 .
	\]
	Since \(\sum_{k=k_0}^{\infty}h_k^2<\infty\), the martingale
	\(\{\bm M_n\}\) is bounded in \(L^2\). By the martingale convergence
	theorem, it converges almost surely. Therefore,
	\begin{equation}
		\lim_{n\to\infty}
		\sup_{m\ge n}
		\left\|
		\sum_{k=n}^{m}h_k\bm Z_0^{(k)}
		\right\|
		=0
		\quad \text{a.s.}
		\label{eq:proof_noise_remainder}
	\end{equation}

	Combining \eqref{eq:proof_sa_form}, \eqref{eq:proof_error_small}, and
	\eqref{eq:proof_noise_remainder}, the resulting late-stage recursion is an Euler
	discretization of the autonomous ODE
	\[
	\frac{d\bm Z(s)}{ds}=\mathcal G(\bm Z(s))
	\]
	up to perturbations whose cumulative effect vanishes on every finite
	algorithmic-time interval. By the ODE method for stochastic approximation,
	the piecewise linear interpolation \(\widetilde{\bm Z}(s)\) is therefore
	an asymptotic pseudo-trajectory of this ODE, which is exactly
	\eqref{eq:asymptotic_ode}. This completes the proof.

	\bibliographystyle{IEEEtran}
	\bibliography{IEEEabrv,StringDefinitions,JCgroup,ChenBibCV}

\end{document}